\documentclass[10pt,a4paper]{article}
\usepackage{amsthm}
\usepackage{amsmath}
\usepackage{amssymb}
\usepackage{amsfonts}
\usepackage{subfigure}
\usepackage{algorithm}
\usepackage{algorithmic}
\usepackage{color}
\usepackage{url}

\usepackage{graphicx}
\usepackage{booktabs}
\usepackage{multirow}

\usepackage[a4paper,
            margin=2 cm]{geometry}

\usepackage{url}

\usepackage{titlesec}
\usepackage{enumitem}

\usepackage{indentfirst}

\usepackage{xcolor}

\usepackage[unicode]{hyperref}
\hypersetup{
    colorlinks = true,
    linkcolor= red,
    citecolor= blue,
    urlcolor=blue
}

\DeclareMathOperator*{\ori}{ori}
\DeclareMathOperator*{\ter}{ter}
\DeclareMathOperator*{\rev}{rev}
\DeclareMathOperator*{\sgn}{sgn}
\DeclareMathOperator*{\TSP}{TSP}
\DeclareMathOperator*{\Diam}{Diam}
\DeclareMathOperator*{\pos}{pos}
\DeclareMathOperator{\mt}{mt}
\DeclareMathOperator*{\cost}{cost}

\theoremstyle{plain}
\newtheorem{theo}{Theorem}[section]
\newtheorem{lemma}[theo]{Lemma}
\newtheorem{claim}[theo]{Claim}
\newtheorem{fact}[theo]{Fact}

\theoremstyle{definition}

\newtheorem{remark}[theo]{Remark}

\title{Online Deterministic Minimum Cost Bipartite Matching with Delays on a Line}
\author{Tung-Wei Kuo}
\date{
Department of Computer Science, National Chengchi University, Taiwan\\
twkuo@cs.nccu.edu.tw 
}
\begin{document}
\maketitle

\begin{abstract}
We study the online minimum cost bipartite perfect matching with delays problem. 
In this problem, $m$ servers and $m$ requests arrive over time, and an online algorithm can delay the matching 
between servers and requests by paying the delay cost. The objective is to minimize the total distance and 
delay cost.
When servers and requests lie in a known metric space, there is a randomized $O(\log n)$-competitive algorithm, 
where $n$ is the size of the metric space. 
When the metric space is unknown a priori, Azar and Jacob-Fanani 
proposed a deterministic 
$O\left(\frac{1}{\epsilon}m^{\log\left(\frac{3+\epsilon}{2}\right)}\right)$-competitive 
algorithm for any fixed $\epsilon > 0$. 
This competitive ratio is tight when $n = 1$ 
and becomes $O(m^{0.59})$ for sufficiently small $\epsilon$.

In this paper, we improve upon the result of Azar and Jacob-Fanani for the case where servers and requests 
are on the real line, providing a deterministic $\tilde{O}(m^{0.5})$-competitive algorithm. 
Our algorithm is based on the Robust Matching (RM) algorithm proposed by Raghvendra for the 
minimum cost bipartite perfect matching problem.  
In this problem, delay is not allowed, and all servers arrive in the beginning. 
When a request arrives, the RM algorithm immediately 
matches the request to a free server based on the request's minimum 
$t$-net-cost augmenting path,  
where $t > 1$ is a constant.  
In our algorithm, we delay the matching of a 
request until its waiting time  
exceeds its minimum $t$-net-cost divided by $t$. 
\end{abstract}
\clearpage

\section{Introduction}
Consider an online gaming platform where players are paired for gameplay. 
To improve the gaming experience, players with similar skill ratings should be matched together. 
However, when a new player joins, a suitable matching may not be available immediately. 
In this case, it is essential to delay the matching process, in the hope 
of finding a better matching in the near future. 
Clearly, the waiting time and the similarity between matched players should be considered 
jointly, and a natural approach is to minimize the sum of both terms.

The above problem is captured by the Minimum cost Perfect Matching with Delays (MPMD) 
problem~\cite{emek2016online}. 
In the MPMD problem, demands arrive over time, and their similarities are modeled by 
a metric space. When a demand arrives, an online algorithm has the option to postpone 
the matching process by incurring a delay cost. 
The objective is to minimize the sum of the total delay time (i.e., delay cost) 
and the total distance between matched demands in the metric space (i.e., distance cost).   

In numerous matching applications, entities can only be matched if they belong to different types 
(e.g., teacher-student, donor-donee, buyer-seller, and driver-passenger). 
These binary classifications motivate the Minimum cost Bipartite Perfect 
Matching with Delays (MBPMD) problem~\cite{azar2017polylogarithmic, ashlagi2017min}. 
In the MBPMD problem, there are two types of demands, servers and requests. 
An online algorithm has to match each request to a server. 
Like the MPMD problem, the objective is to minimize the total distance and delay cost. 

In this paper, we study the MBPMD problem on a line, where requests and servers are positioned 
on the real line. For example, skiers (requests) should be matched to skis (servers) of approximately 
their height~\cite{antoniadis2014competitive}. Another example is matching buyers (requests) and sellers (servers) based on their stated 
prices. In these examples, requests and servers are represented as numbers on the real line, corresponding to 
heights or prices. 
We analyze our algorithm using the standard notion of competitive ratio. In particular, 
an online algorithm is said to be $c$-competitive ($c \geq 1$) 
if for any input, the cost of the algorithm is at most $c$ times the cost of the optimal offline algorithm.

\paragraph*{Background.}
For the MPMD problem, Emek et al. proposed a randomized
$O(\log^2 n + \log \Delta)$-competitive algorithm~\cite{emek2016online}, 
where $n$ is the number of points in the metric space and $\Delta$ is 
the aspect ratio of the metric space.  
Azar et al. then proposed a randomized $O(\log n)$-competitive algorithm, 
and proved that the competitive ratio for any randomized algorithm is 
$\Omega(\sqrt{\log n})$~\cite{azar2017polylogarithmic}. 
Ashlagi et al. further improved this lower bound to
$\Omega\left(\frac{\log n}{\log \log n}\right)$~\cite{ashlagi2017min}. 
All the above randomized algorithms used the celebrated result of 
Fakcharoenphol et al.~\cite{fakcharoenphol2003tight} to transform the original 
metric space into a distribution over Hierarchically Separated Trees (HSTs). 
As a result, these algorithms need to know the metric space in advance. 

The algorithms in~\cite{emek2016online} and \cite{azar2017polylogarithmic} 
are randomized. For offline problems, we can repeatedly execute a randomized 
algorithm until we find a satisfactory solution. However, for online problems, 
we can only execute an algorithm once, and the output cannot be changed. 
Thus, a more robust approach for online problems is to design deterministic algorithms. 

For the MPMD problem, Bienkowski et al. first proposed a deterministic 
$O(m^{2.46})$-competitive algorithm, where $m$ is the number of demands to 
be matched~\cite{bienkowski2017match}. Bienkowski et al. then proposed 
a deterministic $O(m)$-competitive algorithm~\cite{bienkowski2018primal}. 
Finally, Azar and Jacob-Fanani proposed a deterministic 
$O(\frac{1}{\epsilon}m^{\log(\frac{3+\epsilon}{2})})$-competitive 
algorithm for any fixed $\epsilon > 0$~\cite{azar2020deterministic}. 
For small enough $\epsilon$, the competitive ratio becomes $O(m^{0.59})$.  
Unlike the previous randomized algorithms, the above three deterministic 
algorithms do not need to know the metric space in advance. 

When the metric space is a tree, Azar et al. also proposed a deterministic 
$O(n)$-competitive algorithm in~\cite{azar2017polylogarithmic}. 
Moreover, when $n = 2$, Emek et al. proposed a deterministic 3-competitive 
algorithm, and proved that 3 is the best possible 
competitive ratio~\cite{emek2019minimum}.

For the MBPMD problem, Azar et al. first proposed a randomized  
$O(\log n)$-competitive algorithm, and proved that any randomized algorithm 
has a competitive ratio of 
$\Omega(\log^{1/3} n)$~\cite{azar2017polylogarithmic}.
This lower bound is further improved to 
$\Omega\left(\sqrt{\frac{\log n}{\log \log n}}\right)$~\cite{ashlagi2017min}. 
For deterministic algorithms, Bienkowski et al. first proposed an 
$O(m)$-competitive algorithm~\cite{bienkowski2018primal} for the MBPMD problem.  
Azar and Jacob-Fanani then proposed an 
$O(\frac{1}{\epsilon}m^{\log(\frac{3+\epsilon}{2})})$-competitive 
algorithm for any fixed $\epsilon > 0$~\cite{azar2020deterministic}. 
For small enough $\epsilon$, the competitive ratio becomes $O(m^{0.59})$. 
All the above algorithms are based on the algorithms for the 
MPMD problem. Moreover, the competitive ratios in~\cite{bienkowski2018primal} 
and \cite{azar2020deterministic} are tight when the metric space is a 
line.\footnote{In fact, 
the competitive ratios in \cite{bienkowski2018primal} 
and \cite{azar2020deterministic} are tight when $n = 2$ and 
$n = 1$, respectively.} In summary, prior to our work,  
the best known competitive ratio for the deterministic MBPMD problem on a line was $O(m^{0.59})$.

\paragraph*{Our Contribution and Techniques.}
In this paper, we introduce a deterministic $\tilde{O}(m^{0.5})$-competitive 
algorithm for the MBPMD problem on a line, improving upon the $O(m^{0.59})$-competitive 
algorithm of~\cite{azar2020deterministic}. 
Specifically, we have the following result.
\begin{theo}\label{thrm: main}
There is a deterministic $O(\sqrt{m}\log^2 m)$-competitive algorithm 
for the MBPMD problem on a line.
\end{theo}

Our algorithm is based on the Robust Matching (RM) algorithm proposed by 
Raghvendra for the online Minimum cost Bipartite Perfect Matching (MBPM) 
problem~\cite{raghvendra2016robust}. In the MBPM problem, all servers arrive in the beginning, 
and an online algorithm must match a request immediately after it arrives. 
The objective is to minimize the total distance cost of the matching. 
Nayyar and Raghvendra~\cite{nayyar2017input} proved that the competitive ratio of the RM algorithm 
for any $d$-dimensional Euclidean metric space is $O(n^{1-1/d}\log^2 n)$. 
Raghvendra further proved that for one-dimensional Euclidean metric space, RM algorithm is 
$O(\log n)$-competitive~\cite{raghvendra2018optimal}.
From a bird's eye view, RM algorithm maintains an offline matching 
$M^{OFF}$ and an online matching $M^{RM}$, which is the real output matching. 
When request $r_i$ arrives, RM algorithm computes an $M^{OFF}$-augmenting path $P_i$ from  
$r_i$ to some free server $s_j$. RM algorithm then adds $(r_i, s_j)$ to $M^{RM}$ and 
augments $M^{OFF}$ by $P_i$. 

Specifically, $P_i$ is such that minimizes the 
\textit{$\gamma$-net-cost}\footnote{In~\cite{raghvendra2016robust}, this cost is referred to as the $t$-net-cost. 
However, because $t$ denotes time in this paper, we change $t$-net-cost to $\gamma$-net-cost.} 
among all $M^{OFF}$-augmenting paths from $r_i$ to a free server.
When $M^{OFF}$ is augmented by a path $P$, edges shared by $P$ and $M^{OFF}$ are removed from $M^{OFF}$, 
and other edges in $P$ are added to $M^{OFF}$. 
The $\gamma$-net-cost of $P$ is the total distance of the edges added to $M^{OFF}$ 
(multiplied by $\gamma$) minus the total distance of the edges removed from $M^{OFF}$. 

There are two major differences between the MBPM problem and the MBPMD problem. 
Specifically, in the MBPM problem considered by the RM algorithm: 
\begin{enumerate}
\item All servers arrive in the beginning.
\item Once a request arrives, the request must be matched immediately. Thus, 
      the objective function does not consider delay cost. 
\end{enumerate}

To address the above differences, we first introduce a Moving Virtual (MV) server $\widetilde{s}_i$ 
for every request $r_i$.  
Specifically, we consider the Time-Augmented (TA) plane~\cite{bienkowski2017match,azar2020deterministic} 
that adds the time axis to the original one-dimensional space. 
Thus, the TA plane is two-dimensional. 
When $r_i$ arrives, $\widetilde{s}_i$ and $r_i$ are at the same point in the TA plane, 
with the time-coordinate being $r_i$'s arrival time.  
The time-coordinate of $r_i$ is fixed, while 
the time-coordinate of $\widetilde{s}_i$ is always the current time. 
Thus, the distance between $r_i$ and $\widetilde{s}_i$ in the TA plane is always $r_i$'s current waiting time. 
For any real server, its time-coordinate is fixed at its arrival time.

For each request $r_i$, our algorithm maintains two $M^{OFF}$-augmenting paths, 
a \textit{real} augmenting path $P_i$ 
and a \textit{virtual} augmenting path $\widetilde{P}_i$. 
$P_i$ is such that minimizes the $\gamma$-net-cost among all augmenting paths from $r_i$ 
to a real free server in the TA plane,  
and $\widetilde{P}_i$ is such that minimizes the $\gamma$-net-cost among all augmenting paths from 
$r_i$ to an MV server in the TA plane, with the last edge connecting a request $r_p$ 
(possibly different from $r_i$) to $r_p$'s MV server $\widetilde{s}_p$. 
 
Initially, the virtual minimum $\gamma$-net-cost (i.e., the $\gamma$-net-cost of $\widetilde{P}_i$)  
is zero (since $\widetilde{P}_i$ contains only $r_i$ and $\widetilde{s}_i$ initially) and is thus less than or equal to the real minimum $\gamma$-net-cost (i.e., $P_i$'s $\gamma$-net-cost).
After a server arrives or after $M^{OFF}$ is augmented by another 
request's augmenting path, $P_i$ and $\widetilde{P}_i$ may change. 
Moreover, because the distance between $r_p$ and $\widetilde{s}_p$ in the TA plane increases over time, 
the virtual minimum $\gamma$-net-cost increases over time. 
When the virtual minimum $\gamma$-net-cost is greater than or equal to 
the real minimum $\gamma$-net-cost, 
our algorithm matches $r_i$ to the endpoint server of $P_i$ and augments $M^{OFF}$ by $P_i$. 

In Section~\ref{sec: simple}, we show that the algorithm can be greatly simplified: 
we match $r_i$ when its waiting time is
greater than or equal to the real minimum $\gamma$-net-cost divided by $\gamma$. 
Thus, we no longer need MV servers in our algorithm. 
To this end, we prove that the virtual minimum $\gamma$-net-cost is always $\gamma$ times $r_i$'s 
waiting time (Eq.~\eqref{eq: v}).
Nevertheless, MV servers facilitate the analysis of our algorithm in the following senses:
\begin{enumerate}
\item In this paper, we upper bound $r_i$'s delay cost by its real minimum $\gamma$-net-cost. 
Thus, we have to show that the real minimum $\gamma$-net-cost cannot suddenly drop below the 
waiting time, even when a new server arrives. In our proof, we replace the new server 
with an MV server to create a virtual augmenting path $\widetilde{P}$ whose $\gamma$-net-cost is a 
lower bound of the real minimum $\gamma$-net-cost. We then lower bound $\widetilde{P}$'s $\gamma$-net-cost by  
the virtual minimum $\gamma$-net-cost and Eq.~\eqref{eq: v}.

\item When our algorithm matches $r_i$ at time $t$, the optimal solution may match $r_i$ 
to some server $s$ that arrives after time $t$. In our proof, we replace this future server
$s$ with an MV server (which creates a virtual augmenting path) to derive lower bounds for the optimal cost. 
\end{enumerate}
In~\cite{raghvendra2016robust}, it has been shown that the total distance cost 
can be upper bounded by the total $\gamma$-net-cost. Because we further upper bound 
the delay cost by the  $\gamma$-net-cost, our algorithm's total cost
is upper bounded by the total $\gamma$-net-cost. 
To prove Theorem~\ref{thrm: main}, we then use the techniques in~\cite{nayyar2017input} to relate 
the total $\gamma$-net-cost to the optimal cost in the TA plane (recall that the TA plane is two-dimensional).

The main challenge in our analysis is to prove that the real minimum $\gamma$-net-cost 
cannot decrease after $M^{OFF}$ is augmented by another request 
(Lemma~\ref{lemma: AU}), so that we can upper bound 
the delay cost by the real minimum $\gamma$-net-cost. To this end, we partition $P_i$ and 
derive lower bounds for the $\gamma$-net-cost of $P_i$'s subpaths.

\subsection{Other Related Work}
Without considering waiting times (and thus every request must be matched 
immediately upon arrival), there has been a considerable amount of research in 
the literature on how to maximize matching weights or minimize matching costs, 
considering different arrival patterns for vertices or edges. 
Relevant literature on these issues can be found in recent years' studies 
(such as \cite{fahrbach2022edge, kaplan2022online, huang2022power, buchbinder2023lossless, doi:10.1137/1.9781611977912.165}), 
or in the excellent survey by Mehta~\cite{mehta2013online}. 
On the other hand, some studies have explored settings where 
recourse is allowed~\cite{matuschke_et_al:LIPIcs.ICALP.2019.82, gupta2020permutation, megow2020online, angelopoulos2020online, doi:10.1137/1.9781611977912.162}. 
In this subsection, we focus on online problems that allow delays.

\paragraph{Poisson Arrival Processes.}
In~\cite{mari2023online}, Mari et al. considered the MPMD problem and assumed that the request arrival process 
follows a Poisson arrival process. Specifically, for each point $v$ in the 
metric space, the interarrival times of requests at $v$ follow an exponential 
distribution, and arrival processes at different points are independent of each 
other. Mari et al. considered a simple greedy algorithm: when the total waiting time 
of two requests exceeds their distance from each other, 
they are immediately matched. While it has been shown that such an algorithm 
has a competitive ratio of $\Omega(m^{0.58})$ in instances designed by Reingold and 
Tarjan~\cite{reingold1981greedy}, Mari et al. proved that when the request 
arrival process is Poisson, the competitive ratio of this simple greedy 
algorithm is $O(1)$.

\paragraph*{Non-Linear Delay Costs.}
In the MPMD problem, delay cost equals waiting time. 
Other studies have considered different forms of delay costs. 
In~\cite{liu2018impatient}, Liu et al. assumed that delay cost is a convex 
function of waiting time. Specifically, if the waiting time is $t$, 
Liu et al. assumed that the delay cost is $t^{\alpha}$, where $\alpha > 1$. 
Liu et al. considered uniform metric space, where the distance
between any two points is the same. They proposed a deterministic $O(n)$-competitive
algorithm. In~\cite{azar2021min}, 
Azar et al. considered the case where the delay cost is a concave function 
of waiting time. They first considered $n=1$ and proposed an 
$O(1)$-competitive deterministic algorithm. Then they considered $n>1$ 
and designed a randomized $O(\log n)$-competitive algorithm based on 
HST. Azar et al. also considered the bipartite variant 
and proposed an $O(1)$-competitive deterministic algorithm (when $n=1$) 
and a randomized $O(\log n)$-competitive algorithm (when $n>1$). 

In~\cite{deryckere2023online}, Deryckere and Umboh similarly considered 
concave functions and designed a deterministic $O(m)$-competitive 
primal-dual algorithm.
Deryckere and Umboh also utilized set delay functions as delay cost. 
Specifically, at each time $t$, the algorithm incurs delay cost 
as a function of the set of unmatched requests. They proposed a 
deterministic $O(2^m)$-competitive algorithm and a randomized 
$O(m^4)$-competitive algorithm. Their approach is based 
on transforming the MPMD problem into a Metrical Task System 
(MTS) \cite{borodin1992optimal} and solving it using MTS algorithms. 
Deryckere and Umboh also proved that for this problem, 
the competitive ratio of any deterministic algorithm is $\Omega(n)$, 
and the competitive ratio of any randomized algorithm is $\Omega(\log n)$.

\paragraph*{Other Matching Problems With Delays.}
In some games, such as poker or mahjong, 
more than two players are needed. Therefore, 
we need to match more than two requests at once. 
\cite{melnyk2021online, kakimura2023deterministic} considered 
this scenario and designed algorithms based on HST and primal-dual 
transformation. Another way of allowing delay is setting deadlines. 
In~\cite{ashlagi2018maximum}, Ashlagi et al. assumed that 
each request can wait for $\delta$ time units upon arrival and not all 
requests need to be matched. The algorithm aims to find the maximum weight 
matching, and Ashlagi et al. proposed an $O(1)$-competitive algorithm.

\paragraph*{Other Online Problems With Delays.}
Many online problems have variants that allow delays. 
For example, in online network design problems, algorithms need to purchase 
network links upon request arrival. In the case where delay is allowed, 
purchased network links may serve multiple requests simultaneously, 
reducing purchasing costs~\cite{azar2019general, azar2020beyond}. 
Another example is online service problems, in which a server can be moved to serve requests. 
In the case where delay is allowed, we can compute the shortest 
path based on multiple requests, reducing server movement 
distance~\cite{azar2017online, touitou2023improved}. 
For these problems, randomized $O(\text{poly}\log n)$-competitive algorithms are first proposed 
based on HST~\cite{azar2019general, azar2017online}. 
Azar and Touitou then proposed deterministic $O(\text{poly}\log n)$-competitive  
algorithms~\cite{azar2020beyond, touitou2023improved}. 
Recently, Azar. et al. considered the list update with delays problem and 
designed a deterministic $O(1)$-competitive algorithm~\cite{azar2024list}.
 
\section{Preliminaries}
Given two sets $A$ and $B$, a matching $M$ between $A$ and $B$ 
is a set of vertex-disjoint edges between $A$ and $B$.  
An element $v$ is said to be \textit{saturated} 
by $M$ (or $M$ saturates $v$) if $v$ is an endpoint of some edge in $M$. 
A matching $M$ between $A$ and $B$ is said to be \textit{perfect} 
if $M$ saturates all elements in $A \cup B$. 

\paragraph*{Problem Definition.}
In the Minimum cost Bipartite Perfect Matching with Delays (MBPMD) problem, 
there is an underlying metric space $\mathcal{M} = (V, d)$.  
$m$ servers and $m$ requests from $\mathcal{M}$ arrive over time. 
Let $R = \{r_1, r_2, \cdots, r_m\}$ and $S = \{s_1, s_2, \cdots, s_m\}$ 
be the request set and the server set, respectively. 
For any $u \in R \cup S$, $a(u)$ denotes $u$'s arrival time 
and $\ell(u) \in V$ denotes $u$'s location in $\mathcal{M}$. 
When $u$ arrives at time $a(u)$, $\ell(u)$ and $u$'s distances to other 
requests and servers that arrive by time $a(u)$ are revealed to the algorithm. 
In this paper, $r$, $r_i$, and $r_p$ always refer to some request in $R$, 
and $s$, $s_j$, and $s_q$ always refer to some server in $S$.

In the MBPMD problem, an online algorithm computes a perfect matching 
$M$ between $R$ and $S$. After a request $r$ arrives, 
the online algorithm can defer the matching of $r$ by paying the delay cost.
If the algorithm matches $r$ and $s$ at time $t$, 
then the delay cost is $(t-a(r)) + (t-a(s))$.  
After $r$ and $s$ are matched, $(r, s)$ is added to matching $M$ 
and cannot be removed from $M$ afterward. In addition, the algorithm 
pays for the distance cost $d(\ell(r), \ell(s))$. 
For simplicity, for any $u_1, u_2 \in R \cup S$, define $d(u_1, u_2) 
= d(\ell(u_1), \ell(u_2))$. 

In this paper, we assume that $\mathcal{M}$ is a line metric. 
Thus, the location of every server and request is 
a point on the real line, 
and the distance between any two elements in 
$R \cup S$ is their distance on the line. 
Specifically, for any $u \in R \cup S$, 
let $\pos(u) \in \mathbb{R}$ be the location of $u$ on the real line. 
Thus, for any $u_1, u_2 \in R \cup S$, $d(u_1, u_2) = |\pos(u_1)-\pos(u_2)|$.

In summary, an online algorithm for the MBPMD problem has to compute 
a perfect matching $M$ between $R$ and $S$. 
For each edge $(r, s)$ in $M$, let $\mt(r,s)$ be the time when $r$ and $s$ are matched. 
Clearly, $\mt(r,s) \geq a(r)$ and $\mt(r,s) \geq a(s)$.  
Given a perfect matching $M$ and a matching time function $\mt$, 
define 
\[
\cost(M, \mt) = \sum_{(r,s) \in M}{\big(|\pos(r)-\pos(s)| 
+ (\mt(r,s)-a(r)) + (\mt(r,s)-a(s)) \big)}
\]
as the total distance and delay cost of a solution $(M, \mt)$. 
The objective is to minimize $\cost(M, \mt)$. 
 
\paragraph*{Augmenting Paths.}
Given a matching $M$, an $M$-alternating path is a path that alternates betweens edges in $M$ 
and edges not in $M$. An $M$-alternating path $P$ is said to be an $M$-augmenting path if  
both endpoints of $P$ are not saturated by $M$. For any path $P$, we use $E(P)$ 
to denote the set of undirected edges in $P$. For any two sets $A$ and $B$, 
define $A \oplus B = (A \setminus B) \cup (B \setminus A)$ as the symmetric difference between $A$ and $B$. 
Observe that for any $M$-augmenting path $P$, $M \oplus E(P)$ is a matching of size $|M| + 1$. 

While augmenting paths are typically considered as undirected paths, for the sake of convenience, 
we often view an augmenting path as a directed path from an unsaturated request 
$r$ to an unsaturated server $s$. 
We often refer to an augmenting path by the natural sequence of its vertices. 
Specifically, an augmenting path $P$ that originates at $r$ and terminates 
at $s$ can be written in the form of 
$P = r'_1s'_1r'_2s'_2\cdots r'_{\ell}s'_{\ell}$ for some ${\ell} \geq 1$, 
where $r'_1 = r, s'_{\ell} = s$, and  $r'_k \in R, s'_k \in S$ for any $1 \leq k \leq {\ell}$.

Because augmenting paths are directed, edges in augmenting paths can also be viewed as directed edges.
We denote by $\overrightarrow{u,v}$ a directed edge from $u$ to $v$. 
For any augmenting path $P$, we use $\overrightarrow{E}(P)$ to denote the set of directed edges in $P$, 
(i.e., $\overrightarrow{E}(P) = \{\overrightarrow{r'_1,s'_1}, \overrightarrow{s'_1,r'_2}, 
\overrightarrow{r'_2,s'_2}, \cdots, \overrightarrow{r'_{\ell},s'_{\ell}}\}$). 
For any directed edge $\overrightarrow{u,v}$, we say that $\overrightarrow{u,v}$ is in an augmenting path $P$ 
if $\overrightarrow{u,v} \in \overrightarrow{E}(P)$. 
We have the following simple fact. Recall that in this paper, $r$ always refers to a request in $R$ 
and $s$ always refers to a server in $S$.
\begin{fact}\label{fact: s}
Let $M$ be any matching between $R$ and $S$. 
Let $P$ be any $M$-augmenting path that originates at a request and terminates at a server. 
Let $M^{aug} = M \oplus E(P)$. Then the following statements hold: 
\begin{enumerate}
\item For any $\overrightarrow{r, s} \in \overrightarrow{E}(P)$, $(r, s) \notin M$ and $(r, s) \in M^{aug}$. 
\item For any $\overrightarrow{s, r} \in \overrightarrow{E}(P)$, $(r, s) \in M$ and $(r, s) \notin M^{aug}$. 
\end{enumerate}
\end{fact}
In other words, for any directed edge from a request to a server in $P$, 
it is added to $M^{aug}$; 
for any directed edge from a server to a request in $P$, 
it is not in $M^{aug}$. 

We can then express the $\gamma$-net-cost, which was first introduced by Raghvendra~\cite{raghvendra2016robust}, 
based on the directions of edges.
For any $\gamma > 1$ and any augmenting path $P$, 
the $\gamma$-net-cost of $P$ is defined as 
\[
\gamma \left(  \sum_{\overrightarrow{r,s} \in \overrightarrow{E}(P)}{d(r,s)}     \right)
                   -\sum_{\overrightarrow{s,r} \in \overrightarrow{E}(P)}{d(r,s)}.        
\]

\paragraph*{Time Augmented Plane.} 
Throughout this paper, we use $M^{OPT}$ to denote the optimal matching. 
Observe that for any $(r,s) \in M^{OPT}$, the optimal solution must match $r$ and $s$ 
at time $\max{(a(r), a(s))}$. 
Thus, the optimal cost can be written as
\begin{equation}\label{eq: OPTD}
\sum_{(r,s) \in M^{OPT}}{\big(|\pos(r)-\pos(s)| + |a(r)-a(s)|\big)}.
\end{equation}
The above optimal cost suggests that we can view 
$S \cup R$ as a set of points 
in an $xy$-plane, where each point $v \in S \cup R$ 
has $x$-coordinate $\pos(v)$ and $y$-coordinate $a(v)$. 
The $y$-axis can also be viewed as the time axis. 
We call such an $xy$-plane the \textit{Time Augmented (TA) plane}.
Observe that $|\pos(r)-\pos(s)| + |a(r)-a(s)|$ is the Manhattan distance 
between $r$ and $s$ in the TA plane. TA planes are also used 
in~\cite{bienkowski2017match,azar2020deterministic}.

\paragraph*{Paper Organization.} 
In Section~\ref{sec: algo}, we present our algorithm with MV servers for the MBPMD problem on a line 
and state some basic properties. In Section~\ref{sec: simple}, we remove MV servers from our algorithm and 
prove Theorem~\ref{thrm: main}.
\section{An Online Matching Algorithm with Moving Virtual Servers}\label{sec: algo}
In this section, we describe our algorithm, which introduces Moving Virtual (MV) servers 
into the RM algorithm. We thus call our algorithm the Virtual RM (VRM) algorithm. 
In the next section, we will present a simplified version of the algorithm without MV servers.
\subsection{Moving Virtual Servers}\label{subsec: MV}
In the VRM algorithm, whenever a request $r_i$ arrives, the algorithm creates 
a Moving Virtual (MV) server $\widetilde{s}_i$, and sets $a(\widetilde{s}_i) = a(r_i)$. 
Moreover, $\pos(r_i) = \pos(\widetilde{s}_i)$. Therefore, we can also view 
$\widetilde{s}_i$ as a point in the TA plane. 
$\widetilde{s}_i$ moves upward in the TA plane.
Specifically, at time $t \geq a(r_i)$, the $y$-coordinate of $\widetilde{s}_i$ in the TA plane is $t$. 
A simple property is that the distance between an unmatched request 
and its MV server in the TA plane is always the request's current waiting time. 
To differentiate between servers in $S$ and MV servers, we call servers in $S$ the \textbf{real} servers. 
Our algorithm never matches a request to an MV server.

For any $u_1, u_2 \in R \cup S$, the distance between $u_1$ and $u_2$ in the TA plane, 
denoted by $D(u_1, u_2)$, 
is defined as 
\[
D(u_1,u_2) = |\pos(u_1) - \pos(u_2)| + |a(u_1) - a(u_2)|.
\]
For any $r_i \in R$, the distance between $r_i$ and its MV server $\widetilde{s}_i$ 
in the TA plane at time $t$, denoted by $D_t(r_i, \widetilde{s}_i)$, is defined as 
\[
D_t(r_i, \widetilde{s}_i) = t-a(r_i).
\]
An important observation is that if at time $t$, $r_i$ already arrives but a server $s$ has not 
arrived yet (i.e., $a(r_i) \leq t < a(s)$), then 
\begin{equation}\label{eq: ob}
D(r_i, s) \geq a(s) - a(r_i) > t-a(r_i) = D_t(r_i, \widetilde{s}_i).
\end{equation}

\subsection{The VRM Algorithm}\label{subsec: phi}
Like the RM algorithm, the VRM algorithm maintains an offline matching 
$M^{OFF}$ and an online matching $M^{VRM}$, which is the real output matching. 
Unlike the RM algorithm, the VRM algorithm needs to decide the matching time 
for each edge $(r,s) \in M^{VRM}$, denoted by $\mt^{VRM}(r,s)$.
All the servers in $M^{OFF}$ and $M^{VRM}$ are real, 
and these two matchings saturate the same set of servers and requests. 
A real server or a request is said to be \textbf{free} 
if it has arrived but not yet matched by $M^{OFF}$.  
Initially, both matchings are empty. 

We consider two types of augmenting paths, real and virtual, in the TA plane. 
An augmenting path $P$ is \textbf{real} if all servers in $P$ are real. 
An augmenting path $P$ is \textbf{virtual} if the last directed edge of $P$ is from some request $r_p$ 
to $r_p$'s MV server (i.e., $\overrightarrow{r_p, \widetilde{s}_p}$), 
and all the other servers in $P$ are real. 
For any real $M^{OFF}$-augmenting path $P$, define the $\gamma$-net-cost of $P$ in the TA plane, 
denoted by $\varphi_{\gamma}(P)$, as
\[
\varphi_{\gamma}(P) 
         = \gamma \left(  \sum_{\overrightarrow{r, s} \in \overrightarrow{E}(P)}{D(r,s)} \right)
                   -\sum_{\overrightarrow{s, r} \in \overrightarrow{E}(P)}{D(r,s)}.       
\]
The $\gamma$-net-cost of virtual augmenting paths is defined similarly. 
The only difference is that the distance function of the last directed edge is $D_t$ instead of $D$. 
Specifically, for any virtual $M^{OFF}$-augmenting path $\widetilde{P}$ that terminates
at MV server $\widetilde{s}_p$ and any time $t \geq a(r_p)$, 
define the $\gamma$-net-cost of $\widetilde{P}$ at time $t$ in the TA plane, 
denoted by $\varphi_{\gamma,t}(\widetilde{P})$, as
\[
\varphi_{\gamma,t}(\widetilde{P}) 
 = \gamma \left( D_t(r_p, \widetilde{s}_p) + \sum_{\overrightarrow{r, s} \in \overrightarrow{E}(\widetilde{P})}{D(r,s)} \right)
  - \sum_{\overrightarrow{s, r} \in \overrightarrow{E}(\widetilde{P})}{D(r,s)}.                      
\]
We stress that in the above definition, $s$ is a real server.
In this paper, we assume $\gamma = 3$ and 
drop the subscript $\gamma$ in $\varphi_{\gamma}$ and $\varphi_{\gamma,t}$ if it is clear from the context.  

\paragraph*{Description of the Algorithm.}
After a request $r_i$ arrives, the VRM algorithm maintains a real 
$M^{OFF}$-augmenting path $P_i$
and a virtual $M^{OFF}$-augmenting path $\widetilde{P}_i$. 
Specifically, $P_i$ is such that minimizes the $\gamma$-net-cost 
among all real $M^{OFF}$-augmenting paths that originate at $r_i$. 
On the other hand, for any time $t$, $\widetilde{P}_i$ is such that 
minimizes the $\gamma$-net-cost among all virtual $M^{OFF}$-augmenting paths 
that originate at $r_i$ at time $t$.   
We call $P_i$ (respectively, $\widetilde{P}_i$) 
the \textbf{real minimum augmenting path} (respectively, \textbf{virtual minimum augmenting path}) of $r_i$. 
Note that in the absence of free servers, $P_i$ does not exist. 
If so, we assume $\varphi(P_i) = \infty$. 

Fix an offline matching $M^{OFF}$. 
Because $\varphi_{t}(\widetilde{P}_i)$ increases as $t$ increases, 
$\varphi_{t}(\widetilde{P}_i)$ exceeds $\varphi(P_i)$ eventually. 
A request $r_i$ is said to be \textbf{ready} at time $t$ if it is free and
$\varphi_{t}(\widetilde{P}_i) \geq \varphi(P_i)$. 
For any directed path $P$, denote by $\ori(P)$ and $\ter(P)$ as the first and last vertices in $P$, 
respectively. 
Whenever some request $r_i$ is ready at time $t$, 
we first augment $M^{OFF}$ by setting $M^{OFF} \gets M^{OFF} \oplus E(P_{i})$. 
We then add $(r_i, \ter(P_i))$ to $M^{VRM}$, and set $\mt^{VRM}(r_{i}, \ter(P_i)) = t$. 
Finally, we update all free requests' real and virtual minimum augmenting paths 
(since $M^{OFF}$ is changed).\footnote{If multiple requests become 
ready at the same time, 
we choose one of them arbitrarily and process it as described above. 
Other ready requests' real and virtual minimum augmenting paths  
will be updated according to the new offline matching. As a result, some ready 
requests may not be ready after the update.}

In the following, we first explain the subroutine that computes $P_i$ and $\widetilde{P}_i$ 
at any time $t$ (Section~\ref{subsubsec: algoPi}). We then discuss the timings for 
computing $P_i$ and $\widetilde{P}_i$ (Section~\ref{subsubsec: timing}). 

\subsubsection{Computing Minimum Augmenting Paths}\label{subsubsec: algoPi} 
The subroutine for computing $P_i$ and $\widetilde{P}_i$ basically follows that in~\cite{raghvendra2016robust}. 
Let $\widetilde{S} = \{\widetilde{s}_1, \widetilde{s}_2, \cdots, \widetilde{s}_m\}$.
Our algorithm maintains dual variables $z(\cdot)$ for $R \cup S \cup \widetilde{S}$ 
and the following invariants:
\begin{equation}\label{eq: dualrelaxed}
z(r)+z(s) \leq \gamma D(r,s), \forall r \in R, s \in S. 
\tag{I1}
\end{equation}

\begin{equation}\label{eq: dualrelaxed2}
z(r_p)+z(\widetilde{s}_p) \leq \gamma D_t(r_p, \widetilde{s}_p), 
\forall r_p \in R, t \geq a(r_p).
\tag{I2}
\end{equation}

\begin{equation}\label{eq: dual0}
z(v) = 0, \forall v \in \widetilde{S} \cup \{u | u \in R\cup S, u \text{ is not saturated by } M^{OFF}\}. 
\tag{I3}
\end{equation}

\begin{equation}\label{eq: dualtight}
z(r)+z(s) = D(r,s), \forall (r,s) \in M^{OFF}. 
\tag{I4}
\end{equation}

By Invariant~\eqref{eq: dual0}, all the dual variables are zero initially. 
It is easy to see that all invariants hold initially. 
$z(\cdot)$ is only updated when $M^{OFF}$ is augmented. 
The subroutine for updating $z(\cdot)$ is discussed in Section~\ref{subsubsec: dual}.
Fix a time $t$, a free request $r_i$, an offline matching $M^{OFF}$, and $z(\cdot)$, 
we next explain how to compute $P_i$ and $\widetilde{P}_i$. 

Let $R_{sat}$ be the set of requests that are saturated by $M^{OFF}$, 
and $\widetilde{S}_{sat}$ be the set of MV servers of the requests in $R_{sat}$. 
Let $S_t$ be the set of servers that arrive by time $t$.
We construct an edge-weighted directed bipartite graph $G_{i,t}$ 
with partite sets $R_{sat} \cup \{r_i\}$ and $S_t \cup \widetilde{S}_{sat} \cup \{\widetilde{s}_i\}$. 
The edge weight is the edge's slack with respect to dual variables $z(\cdot)$. 
$G_{i,t}$ is called the \textbf{slack graph} of $r_i$ at time $t$ and is constructed as follows:

\begin{enumerate}
\item For every server $s$ in $S_t$ and every request $r$ in $R_{sat}$, 
if $(s,r) \in M^{OFF}$, 
we add to $G_{i,t}$ a directed edge $\overrightarrow{s,r}$ 
with edge weight $sl(\overrightarrow{s, r}) = 0$.
\item For every server $s$ in $S_t$ and every request $r$ in $R_{sat} \cup \{r_i\}$, 
if $(r,s) \notin M^{OFF}$, 
we add to $G_{i,t}$ a directed edge $\overrightarrow{r,s}$ 
with edge weight $sl(\overrightarrow{r, s}) = \gamma D(r,s) - (z(r)+z(s))$. 
\item For every request $r_p$ in $R_{sat} \cup \{r_i\}$, 
we add to $G_{i,t}$ a directed edge $\overrightarrow{r_p,\widetilde{s}_p}$ 
with edge weight 
$sl(\overrightarrow{r_p,\widetilde{s}_p}) = \gamma(D_t(r_p,\widetilde{s}_p)) - (z(r_p)+z(\widetilde{s}_p))$.
\end{enumerate}

We then set $P_i$ as the shortest path from $r_i$ to the set of real free server in $G_{i,t}$. 
This can be done by first computing all the shortest paths from $r_i$ to each free real server in $G_{i,t}$ 
and then outputting the shortest one among these paths. 
Similarly, we set $\widetilde{P}_i$ as the shortest path from $r_i$ to the set of MV servers in $G_{i,t}$. 
In Appendix~\ref{appendix: correctness}, we prove that this subroutine indeed 
computes the real and virtual minimum augmenting paths. 
In particular, it can be shown that for any augmenting path $P$, $P$'s weight on $G_{i,t}$ 
equals $P$'s $\gamma$-net-cost.
The proof is similar to that in~\cite{raghvendra2016robust}.

\subsubsection{Timings for Updating $P_i$ and $\widetilde{P}_i$}\label{subsubsec: timing}
For each free request $r_i$, we use the subroutine in Section~\ref{subsubsec: algoPi} to 
compute $P_i$ and $\widetilde{P}_i$ when $r_i$ arrives or whenever 
one of the following events occurs:
\begin{description}
\item[Event SA:] A server arrives.
\item[Event AU:] $M^{OFF}$ is augmented by another request's augmenting path.
\end{description}
When Event~SA occurs, we only need to update $P_i$, 
as $\widetilde{P}_i$ cannot change due to the arrival of a new server. 

Observe that all virtual augmenting paths increase the $\gamma$-net-cost by the same speed. 
Specifically, for any virtual augmenting path $\widetilde{P}$ at time $t$ and some future time $t' > t$, 
we have $\varphi_{t'}(\widetilde{P}) = \varphi_{t}(\widetilde{P}) + \gamma (t'-t)$. 
Thus, if Event AU does not occur (and thus $M^{OFF}$ does not change), 
then $r_i$'s virtual minimum augmenting path cannot change. 
Therefore, we only need to update $\widetilde{P}_i$ when Event AU occurs. 

Whenever $P_i$ or $\widetilde{P}_i$ is updated at some time $t$, 
we check whether $r_i$ becomes ready 
(i.e., $\varphi_{t}(\widetilde{P}_i) \geq \varphi(P_i)$). 
If not, we compute the following \textbf{ready timing} $t^{rdy}_i$ for $r_i$: 
\[
t^{rdy}_i = t + \frac{\varphi(P_i) - \varphi_{t}(\widetilde{P}_i)}{\gamma}.
\]
At time $t^{rdy}_i$, $r_i$ becomes ready. Note that $P_i$ and $\widetilde{P}_i$ may change before time $t^{rdy}_i$.   
If so, we update $t^{rdy}_i$ again.

\subsubsection{Updating the Dual Variables} \label{subsubsec: dual}
We update the dual variables $z(\cdot)$ whenever $M^{OFF}$ is augmented by some ready 
request $r_i$'s augmenting path $P_i$. 
Let $t$ be the time when $M^{OFF}$ is augmented by $P_i$.
Let $G^{pre} = G_{i, t}$ be the slack graph of $r_i$ right before $M^{OFF}$ is augmented by $P_i$. 
For each vertex $v$ in $G^{pre}$, define $sl(r_i, v)$ as the shortest distance from $r_i$ to $v$ in 
$G^{pre}$. Let $s^* = \ter(P_i)$. We update $z(\cdot)$ in two steps. 

\begin{enumerate}[leftmargin=20mm, label= \textbf{Step }\arabic*:, ref=\arabic*]
\item
    \begin{itemize}
        \item For every request $r$ in $G^{pre}$, if $sl(r_i, r) < sl(r_i, s^*)$, we then set 
              \[
              z(r) \gets z(r) + (sl(r_i, s^*) - sl(r_i, r)).
              \]
        \item For every real server $s$ in $G^{pre}$, if $sl(r_i, s) < sl(r_i, s^*)$, we then set  
              \[
              z(s) \gets z(s) - (sl(r_i, s^*) - sl(r_i, s)).
              \]
    \end{itemize} \label{step1}
\item
    \begin{itemize} 
        \item For every $\overrightarrow{r,s} \in \overrightarrow{E}(P_i)$, 
        we set $z(r) \gets z(r) - (\gamma-1)D(r,s)$. 
    \end{itemize} \label{step2}
\end{enumerate}
Note that for every directed edge $\overrightarrow{r,s}$ considered 
in Step~\ref{step2}, 
$(r,s)$ is added to $M^{OFF}$ after $r_i$ is matched. 

Recall that all the invariants hold initially. 
In Appendix~\ref{appendix: dual}, we prove that all the invariants hold 
after Step~\ref{step2} is executed. The proof is similar to that in~\cite{raghvendra2016robust}. 
One difference is that to prove Invariant~\eqref{eq: dualrelaxed2}, 
we use the property of the VRM algorithm that when $r_i$ is matched at time 
$t$, $\varphi_{t}(\widetilde{P}_i) \geq \varphi(P_i)$ must hold. 
Thus, for any request $r_p$ in $G^{pre}$, 
we have $sl(r_i, r_p) + sl(\overrightarrow{r_p, \widetilde{s}_p}) 
     = sl(r_i, \widetilde{s}_p) \geq \varphi_{t}(\widetilde{P}_i) 
\geq \varphi(P_i) = sl(r_i, s^*)$. 
Therefore, the increase of $z(r_p)$ due to Step~\ref{step1}
(i.e., $sl(r_i, s^*) - sl(r_i, r_p)$) 
is at most the slack between $r_p$ and its MV server.
Another difference lies in the proof of Invariant~\eqref{eq: dualrelaxed} for  
servers that have not arrived yet. Specifically, let $r_p$ be any request 
in $G^{pre}$ and $s$ be any server that is not in $G^{pre}$. 
Because $z(s) = 0$, we only need to prove $\gamma D(r_p,s) \geq z(r_p)$. 
Because $\gamma D(r_p, s) \geq \gamma D_{t}(r_p, \widetilde{s}_p)$ 
(by Eq.~\eqref{eq: ob}),  
the proof then follows from Invariant~\eqref{eq: dualrelaxed2}.

\subsection{Upper Bounding Distance in the TA Plane 
and Nonnegativity of $\gamma$-Net-Costs}
We first state a lemma that upper bounds $D(M^{VRM})$ 
by the sum of the minimum $\gamma$-net-cost over all requests. 
Specifically, let $P^*_i$ be the real minimum augmenting path of $r_i$ 
when $r_i$ is ready and matched by the VRM algorithm. The proof of the following lemma is 
similar to that in~\cite{raghvendra2016robust}, and can be found in Appendix~\ref{appendix: lemma: DUB}. 
For any matching $M$, define $D(M) = \sum_{(r,s) \in M}{D(r,s)}$. 
\begin{lemma}\label{lemma: DUB}
Let $\gamma > 1$ and $M^{VRM}$ be the final online matching output by the VRM algorithm. 
Then 
\[
D(M^{VRM}) \leq \frac{2}{\gamma-1}\sum_{i=1}^{m}{\varphi(P^*_i)}.
\]
\end{lemma}

The following lemmas hold throughout the execution of the VRM algorithm. 
The proofs are similar to those in~\cite{nayyar2017input}, and can be found in 
Appendix~\ref{appendix: properties}.  
We stress that in Lemma~\ref{lemma: realgeq0}, 
$\ter(P)$ can be any server in $S$ regardless of its arrival time. 
\begin{lemma}\label{lemma: realgeq0}
For any real $M^{OFF}$-augmenting path $P$, $\varphi(P) \geq 0$.
\end{lemma}

\begin{lemma}\label{lemma: virtualgeq0}
For any time $t$ and any virtual $M^{OFF}$-augmenting path $\widetilde{P}$ 
at time $t$, $\varphi_{t}(\widetilde{P}) \geq 0$.
\end{lemma}

\section{A Simplified Algorithm and Analysis} \label{sec: simple}
In this section, we present a simplified algorithm that maintains the same $M^{OFF}$ and $M^{VRM}$ 
as the VRM algorithm without MV servers. 
To this end, we prove some properties regarding the change of real and virtual augmenting paths. 
Throughout the execution of the VRM algorithm, 
three types of events may occur: a server arrives (Event SA), a request becomes ready (Event AU), 
and a request arrives (Event RA). 
The VRM algorithm can be described as an event-driven algorithm: 
\begin{enumerate}
\item When Event RA occurs, compute the real and virtual minimum augmenting paths for the new request. 
\item When Event SA occurs, update the real minimum augmenting paths for all free requests.
\item When Event AU occurs, augment $M^{OFF}$ and update $M^{VRM}$ according to the real minimum 
      augmenting path of the request that becomes ready, and then update the real and virtual 
      minimum augmenting paths for all free requests. 
\end{enumerate}
Multiple events may occur simultaneously. 
If so, these events can be processed in any order.

Fix an arbitrary request $r_i$. Assume that starting from time $a(r_i)$ to the time when $r_i$ is matched 
by the VRM algorithm, the algorithm processes events $E_1, E_2, \cdots, E_{\nu}$ in order, 
where $E_1$ is the event that $r_i$ arrives and $E_{\nu}$ is the event that $M^{OFF}$ is augmented 
by $r_i$'s augmenting path (and thus $r_i$ is matched by the VRM algorithm). For any $1 \leq w \leq \nu$, 
let $t_w$ be the time when $E_w$ is processed. 
We introduce the following notations to distinguish the states of the VRM algorithm 
right before an event occurs and right after the event is processed. 
\begin{itemize}
\item Let $P_{i,w}^{pre}$ and $P_{i,w}^{post}$ be $r_i$'s real minimum augmenting paths 
right before $E_w$ occurs and right after $E_w$ is processed, respectively. 
\item Let $\widetilde{P}_{i,w}^{pre}$ and $\widetilde{P}_{i,w}^{post}$ be $r_i$'s virtual minimum 
augmenting paths right before $E_w$ occurs and right after $E_w$ is processed, respectively. 
\end{itemize}
Define $\varphi^{post}_{i,w} = \varphi(P^{post}_{i,w})$ 
and $\varphi^{pre}_{i,w} = \varphi(P^{pre}_{i,w})$. 
Similarly, define 
$\widetilde{\varphi}^{post}_{i,w} = \varphi_{t_w}(\widetilde{P}^{post}_{i,w})$ 
and $\widetilde{\varphi}^{pre}_{i,w} = \varphi_{t_w}(\widetilde{P}^{pre}_{i,w})$.

In this paper, we prove that the following two inequalities  
hold for any $1 \leq w \leq \nu-1$ (Lemma~\ref{lemma: vlerev}).
The first one states that after an event is processed, 
the $\gamma$-net-cost of $\widetilde{P}_i$ is at most that of $P_i$. 
\begin{equation}\label{eq: vlere}
\widetilde{\varphi}^{post}_{i,w} \leq \varphi^{post}_{i,w}.
\end{equation}
The next one states that 
the simplest virtual augmenting path, $r_i\widetilde{s}_i$, 
is always the the best one. As a result, the virtual minimum $\gamma$-net-cost is 
always $\gamma$ times $r_i$'s waiting time.
\begin{equation}\label{eq: v}
\widetilde{\varphi}^{post}_{i,w} = \gamma(t_w-a(r_i)).
\end{equation}

\begin{lemma}\label{lemma: vlerev}
Eq.~\eqref{eq: vlere} and Eq.~\eqref{eq: v} hold for any $1 \leq w \leq \nu-1$.
\end{lemma}

\subsection{Implications of Lemma~\ref{lemma: vlerev}}
\paragraph*{Implication 1: A simplified algorithm.} 
Recall that $r_i$ becomes ready at time $t$ if $\varphi_{t}(\widetilde{P}_i) \geq \varphi(P_i)$, 
and this is the only reason that we need to compute $\widetilde{P}_i$. 
By Eq.~\eqref{eq: v}, to determine whether $r_i$ is ready, 
we only need to compare $\varphi(P_i)$ and $\gamma(t-a(r_i))$. 
Specifically, whenever $P_i$ is updated at some time $t$, 
we check whether $r_i$ becomes ready 
(i.e., $\gamma(t-a(r_i)) \geq \varphi(P_i)$). 
If not, we compute the following \textit{ready timing} 
$t^{rdy}_i$ for $r_i$: 
\[
t^{rdy}_i = t + \frac{\varphi(P_i) - \gamma(t-a(r_i))}{\gamma}.
\]
As a result, in the simplified algorithm, we do not need to compute $\widetilde{P}_i$ explicitly. 

\paragraph*{Implication 2: An upper bound for $\cost(M^{VRM}, \mt^{VRM})$.} 
Observe that for any request $r$ and server $s$ that are 
matched together at time $\mt(r,s)$, 
by the triangle inequality, 
we have 
\[
\mt(r,s)-a(s) \leq \mt(r,s)-a(r) + |a(r)-a(s)|.
\] 
Thus, the distance and delay cost of $(r,s)$ can be upper bounded as follows: 
\begin{align*}
&|\pos(r)-\pos(s)| + \mt(r,s)-a(r) + \mt(r,s)-a(s) \\
\leq &|\pos(r)-\pos(s)| + 2(\mt(r,s)-a(r)) + |a(r)-a(s)|\\
= &D(r,s) + 2(\mt(r,s)-a(r)).
\end{align*}
Thus,
\begin{equation}\label{eq: imp2}
\cost(M^{VRM}, \mt^{VRM}) \leq D(M^{VRM}) 
+ 2\sum_{(r_i, s_j) \in M^{VRM}}{\big(\mt^{VRM}(r_i, s_j)-a(r_i)\big)}.
\end{equation}

By Lemma~\ref{lemma: vlerev} (Eq.~\eqref{eq: vlere}), after $E_{\nu-1}$ is processed, 
$\varphi_{t_{\nu-1}}(\widetilde{P}_i) \leq \varphi(P_i)$. 
Because no event is processed between $E_{\nu-1}$ and $E_{\nu}$, 
when $r_i$ becomes ready and matched by the VRM algorithm at time $t_{\nu}$, we must have  
\[\varphi(P_i) = \varphi_{t_{\nu}}(\widetilde{P}_i) = \widetilde{\varphi}^{pre}_{i,\nu}.\] 
Moreover, we have $\widetilde{\varphi}^{pre}_{i,\nu} 
= \widetilde{\varphi}^{post}_{i,\nu-1} + \gamma(t_{\nu} - t_{\nu-1})
= \gamma(t_{\nu}-a(r_i))$, where the last equality holds due to
Lemma~\ref{lemma: vlerev} (Eq.~\eqref{eq: v}).   
Thus, for any $(r_i, s_j) \in M^{VRM}$, we have  
$\varphi(P^*_i) = \gamma(t_{\nu}-a(r_i)) = \gamma(\mt^{VRM}(r_i, s_j)-a(r_i))$, or equivalently,  
$\mt^{VRM}(r_i, s_j)-a(r_i) = \frac{1}{\gamma} \varphi(P^*_i)$. 
Combining with Eq.~\eqref{eq: imp2}, Lemma~\ref{lemma: DUB}, and $\gamma = 3$, we then have 
\begin{equation}\label{eq: costUB}
\cost(M^{VRM}, \mt^{VRM}) = O(1)\sum_{i=1}^{m}{\varphi(P^*_i)}.
\end{equation}

\subsection{Proof of Theorem~\ref{thrm: main}} 
Recall that by Eq.~\eqref{eq: OPTD}, the optimal cost is $D(M^{OPT})$. 
By Eq.~\eqref{eq: costUB}, to prove Theorem~\ref{thrm: main}, it suffices to show
\begin{equation}\label{eq: maingoal1}
\sum_{i=1}^{m}{\varphi(P^*_i)} = O(\sqrt{m} \log^2 m) D(M^{OPT}).
\end{equation}

The proof is similar to that in~\cite{nayyar2017input}, and we give the proof in 
Appendix~\ref{appendix: eq: maingoal1}. 
The main difference is that in the MBPMD problem, servers arrive over time. 
However, in~\cite{nayyar2017input}, all servers arrive in the beginning.  
In our proof, we replace servers that arrive in the future 
with proper MV servers. 
For example, to relate $\varphi(P^*_i)$ to $D(M^{OPT})$, 
\cite{nayyar2017input} considered $r_i$'s augmenting path $P$ 
in $M^{OFF} \oplus M^{OPT}$, where $M^{OFF}$ is the offline matching 
right before it is augmented by $r_i$'s minimum augmenting path.  
Let $r_p$ be the last request in $P$. 
If $\ter(P)$ arrives in the future, 
we construct a virtual augmenting path $\widetilde{P}$ of $r_i$ 
by replacing $\ter(P)$ with $\widetilde{s}_p$ in $P$. 
Because $\ter(P)$ arrives in the future, 
$\varphi_{t}(\widetilde{P}) \leq \varphi(P)$, where $t$ is the time when 
$r_i$ is matched by the VRM algorithm. 
By the design of the VRM algorithm, $r_i$ becomes ready at time $t$, 
which implies $\varphi(P^*_i) \leq \varphi_{t}(\widetilde{P}) 
\leq \varphi(P) \leq \gamma D(M^{OPT})$.

Another difference is that in~\cite{nayyar2017input}, 
the competitive ratio is input sensitive in the sense that 
the competitive ratio is a function of the server locations. 
However, because servers are not known in advance in the MBPMD problem, 
we do not consider input sensitive analysis, and thus  
our proof is simpler. 

\subsection{Proof of Lemma~\ref{lemma: vlerev}}
We prove Lemma~\ref{lemma: vlerev} by induction on $w$. 
When $w = 1$, $E_w$ is the event that $r_i$ arrives and thus $t_1 = a(r_i)$. 
Clearly, $\varphi_{t_1}(r_i\widetilde{s}_i) = 0$. 
Combining with Lemma~\ref{lemma: virtualgeq0}, 
we then have $\widetilde{\varphi}_{i,1}^{post} = 0$. 
Thus, Eq.~\eqref{eq: v} holds when $w = 1$. 
Moreover, by Lemma~\ref{lemma: realgeq0}, Eq.~\eqref{eq: vlere} also holds when $w = 1$. 

Assume that Eq.~\eqref{eq: vlere} and Eq.~\eqref{eq: v} 
hold for some $w \in \{1,2, \cdots, \nu-2\}$. 
We will prove that Eq.~\eqref{eq: vlere} and Eq.~\eqref{eq: v} hold for $w+1$, which completes the proof. 
Because no event occurs between $E_w$ and $E_{w+1}$, $P_i$ and $\widetilde{P}_i$ 
do not change during $E_w$ and $E_{w+1}$. 
In addition, after processing $E_w$, we must have $t^{rdy}_i \geq t_{w+1}$. 
(Otherwise, if $t^{rdy}_i < t_{w+1}$, then an event (of type AU) 
should occur between $E_w$ and $E_{w+1}$, which leads to a contradiction.) 
Thus, by the induction hypothesis $\widetilde{\varphi}^{post}_{i,w} \leq \varphi^{post}_{i,w}$, 
we then have
\begin{equation}\label{eq: pre}
\widetilde{\varphi}^{pre}_{i,w+1} \leq \varphi^{pre}_{i,w+1}.
\end{equation}
Moreover, by the induction hypothesis $\widetilde{\varphi}^{post}_{i,w} = \gamma(t_w-a(r_i))$, 
we have
\begin{equation}\label{eq: tildephipre}
\widetilde{\varphi}^{pre}_{i,w+1} = \widetilde{\varphi}^{post}_{i,w} + \gamma(t_{w+1}-t_w) 
                                  = \gamma(t_{w+1}-a(r_i)).
\end{equation}

We divide the proof into three cases according to the type of $E_{w+1}$:
\paragraph*{Case 1: Event $E_{w+1}$ is of type RA.} In this case, $r_i$'s real and virtual minimum augmenting paths do not change due to $E_{w+1}$. 
Specifically,
\[
P^{post}_{i,w+1} = P^{pre}_{i,w+1} \text{ and } \varphi^{post}_{i,w+1} = \varphi^{pre}_{i,w+1}
\]
and 
\[
\widetilde{P}^{post}_{i,w+1} = \widetilde{P}^{pre}_{i,w+1} \text{ and } 
\widetilde{\varphi}^{post}_{i,w+1} = \widetilde{\varphi}^{pre}_{i,w+1}
\]
Thus, by Eq.~\eqref{eq: pre} and Eq.~\eqref{eq: tildephipre}, 
Eq.~\eqref{eq: vlere} and Eq.~\eqref{eq: v} hold for $w+1$ in this case. 

\paragraph*{Case 2: Event $E_{w+1}$ is of type SA.} 
Assume that in Event $E_{w+1}$, a server $s$ arrives. 
Because $\widetilde{P}_i$ does not change due to the arrival of a new server, 
we then have 
\begin{equation}\label{eq: case2v}
\widetilde{P}^{post}_{i,w+1} = \widetilde{P}^{pre}_{i,w+1} \text{ and } 
\widetilde{\varphi}^{post}_{i,w+1} = \widetilde{\varphi}^{pre}_{i,w+1}.
\end{equation}
Thus, by Eq.~\eqref{eq: tildephipre}, Eq.~\eqref{eq: v} holds for $w+1$ in this case.
 
Assume that $\overrightarrow{r_p, s_q}$ is the last directed edge in $P^{post}_{i,w+1}$. 
We first consider the subcase where $s_q \neq s$. 
In this subcase, $P_i$ is not affected by the arrival of $s$. 
We then have 
\[
P^{post}_{i,w+1} = P^{pre}_{i,w+1} \text{ and } \varphi^{post}_{i,w+1} = \varphi^{pre}_{i,w+1}.
\]
Thus, by Eq.~\eqref{eq: pre} and Eq.~\eqref{eq: case2v}, Eq.~\eqref{eq: vlere} holds for $w+1$ in this subcase. 

Next, consider the subcase where $s_q = s$. 
Recall that $\overrightarrow{r_p, s_q}$ is the last directed edge in $P^{post}_{i,w+1}$. 
Consider a virtual augmenting path $\widetilde{P}$ 
obtained by replacing $s_q$ with $\widetilde{s}_p$ in $P^{post}_{i,w+1}$. 
Because $D(r_p, s_q) \geq a(s_q) - a(r_p) = t_{w+1} - a(r_p) = D_{t_{w+1}}(r_p, \widetilde{s}_p)$, 
we have 
\[
\varphi^{post}_{i,w+1} \geq \varphi_{t_{w+1}}(\widetilde{P}).
\]
Clearly, $\widetilde{P}$ is a valid virtual augmenting path of $r_i$ after $E_{w+1}$. 
Thus, by the optimality of $\widetilde{\varphi}^{post}_{i,w+1}$, 
we then have 
$\varphi_{t_{w+1}}(\widetilde{P}) \geq \widetilde{\varphi}^{post}_{i,w+1}$. 
Therefore, Eq.~\eqref{eq: vlere} also holds for $w+1$ in this subcase.

\paragraph*{Case 3: Event $E_{w+1}$ is of type AU.} 
The proof of this case follows immediately from  
Eq.~\eqref{eq: pre}, Eq.~\eqref{eq: tildephipre}, and the next lemma. 
\begin{lemma}\label{lemma: AU} 
If Eq.~\eqref{eq: vlere} and Eq.~\eqref{eq: v} hold 
for some $w \in \{1, 2, \cdots, \nu-2\}$, 
and Event $E_{w+1}$ is of type AU, then 
\[
\varphi_{i,w+1}^{post} \geq \varphi_{i,w+1}^{pre} \text{ and }
\widetilde{\varphi}_{i,w+1}^{post} = \widetilde{\varphi}_{i,w+1}^{pre}.
\]
\end{lemma}

\begin{proof}
Let $M^{pre}_{w+1}$ be the offline matching right before $E_{w+1}$ occurs 
and $M^{post}_{w+1}$ be the offline matching right after $E_{w+1}$ is processed. 
Assume that in $E_{w+1}$, the VRM algorithm augments $M^{pre}_{w+1}$ 
by request $r^*$'s minimum augmenting path, denoted by $P^*$. 
Thus, $M^{post}_{w+1} = M^{pre}_{w+1} \oplus E(P^*) 
= (M^{pre}_{w+1} \setminus E(P^*)) \cup (E(P^*) \setminus M^{pre}_{w+1})$. 
Note that $P^{post}_{i,w+1}$ and $\widetilde{P}^{post}_{i,w+1}$ are $M^{post}_{w+1}$-augmenting paths.

Observe that if $P^{post}_{i,w+1}$ does not use any edge in $P^*$ 
(i.e., $E(P^{post}_{i,w+1}) \cap E(P^*) = \varnothing$), 
then $P^{post}_{i,w+1}$ is also a real $M^{pre}_{w+1}$-augmenting path of $r_i$\footnote{A path $P$ 
is said to be $r$'s $M$-augmenting path (or an $M$-augmenting path of $r$) if $P$ is 
an $M$-augmenting path originating at $r$.} (in particular, 
$P^{post}_{i,w+1}$ is a real $(M^{pre}_{w+1} \setminus E(P^*))$-augmenting path). 
By the optimality of $\varphi^{pre}_{i,w+1}$, we have 
$\varphi^{post}_{i,w+1} = \varphi(P^{post}_{i,w+1}) \geq \varphi^{pre}_{i,w+1}$ 
as desired. 
Similarly, if $\widetilde{P}^{post}_{i,w+1}$ does not use any edge in $P^*$, 
then $\widetilde{P}^{post}_{i,w+1}$ is also a virtual $M^{pre}_{w+1}$-augmenting path 
of $r_i$. Thus, 
$\widetilde{\varphi}^{post}_{i,w+1} = \varphi_{t_{w+1}}(\widetilde{P}^{post}_{i,w+1}) 
\geq \widetilde{\varphi}^{pre}_{i,w+1}$ as desired. 
Therefore, we assume $E(P^{post}_{i,w+1}) \cap E(P^*) \neq \varnothing$ 
and $E(\widetilde{P}^{post}_{i,w+1}) \cap E(P^*) \neq \varnothing$. 

To prove $\varphi_{i,w+1}^{post} \geq \varphi_{i,w+1}^{pre}$, 
for any $r_i$'s real $M^{post}_{w+1}$-augmenting path $P$ (e.g., $P^{post}_{i,w+1}$), 
we construct a real $M^{pre}_{w+1}$-augmenting path $P'$ of $r_i$ such that 
\[
\varphi(P') \leq \varphi(P).
\]
As a result, when $P = P^{post}_{i,w+1}$, we have 
$\varphi_{i,w+1}^{pre} \leq \varphi(P') \leq \varphi(P) = \varphi^{post}_{i,w+1}$ 
as desired.  

To prove $\widetilde{\varphi}_{i,w+1}^{post} 
= \widetilde{\varphi}_{i,w+1}^{pre}$, 
we first prove 
$\widetilde{\varphi}_{i,w+1}^{post} \geq \widetilde{\varphi}_{i,w+1}^{pre}$. 
For any $r_i$'s virtual $M^{post}$-augmenting path $\widetilde{P}$ 
(e.g., $\widetilde{P}^{post}_{i,w+1}$), 
we construct a real $M^{pre}$-augmenting path $P''$ of $r_i$ such that 
\[
\varphi(P'') \leq \varphi_{t_{w+1}}(\widetilde{P}).
\] 
Because Eq.~\eqref{eq: vlere} holds for $w$, 
we have $\widetilde{\varphi}^{pre}_{i, w+1} \leq \varphi^{pre}_{i,w+1}$. 
As a result, when $\widetilde{P} = \widetilde{P}^{post}_{i,w+1}$, we have 
\[
\widetilde{\varphi}_{i,w+1}^{pre} \leq \varphi^{pre}_{i,w+1} \leq \varphi(P'') 
\leq \varphi_{t_{w+1}}(\widetilde{P})
= \widetilde{\varphi}^{post}_{i,w+1}.
\]
  
Next, we prove $\widetilde{\varphi}^{post}_{i,w+1} \leq \widetilde{\varphi}_{i,w+1}^{pre}$. 
Because $r_i\widetilde{s}_i$ is a valid virtual augmenting path for $r_i$ after $E_{w+1}$ is processed 
and Eq.~\eqref{eq: v} holds for $w$, we then have 
$\widetilde{\varphi}^{post}_{i,w+1}
\leq \varphi_{t_{w+1}}(r_i\widetilde{s}_i) = \gamma(t_{w+1}-a(r_i)) = \widetilde{\varphi}_{i,w+1}^{pre}$.

\paragraph*{Construction of $P'$ and $P''$.} 
We only give a proof sketch here.
The complete proof can be found in Appendix~\ref{appendix: lemma: AU}. 
The constructions of $P'$ and $P''$ are similar. 
Recall that $P$ is any $r_i$'s real $M^{post}_{w+1}$-augmenting path 
(e.g., $P^{post}_{i,w+1}$). 
Assume that $s$ is the only server shared by both $P$ and $P^*$.  
To construct $P'$, we first traverse $P$ until $s$ is reached, 
after which we traverse $P^*$. We can then prove 
$\varphi(P') \leq \varphi(P)$ by the optimality of $\varphi(P^*)$. 
The main difficulty of the proof is to handle the case 
where $P$ and $P^*$ intersect at multiple servers. 
To this end, we partition $P$ into subpaths at the intersections 
of $P$ and $P^*$, and derive lower bounds of these subpaths' $\gamma$-net-costs.  
\end{proof}

\section{Concluding Remarks}
A natural open question regarding Theorem~\ref{thrm: main} is whether the competitive ratio 
is asymptotically tight. For the MBPM problem, RM algorithm's competitive ratio in~\cite{nayyar2017input}
is almost tight (up to a polylogarithmic factor). 
However, because matching can be delayed in the MBPMD problem, 
the lower bound instance and its analysis in~\cite{nayyar2017input} cannot be directly extended to the MBPMD problem.

\bibliographystyle{alpha}
\bibliography{ref}
\clearpage
\appendix
\section{Missing Proofs in Section~\ref{sec: algo}}
\subsection{Correctness of the Minimum Augmenting Path Subroutine}\label{appendix: correctness}
Recall that the subroutine sets $P_i$ as the shortest path from $r_i$ to the set of real free servers 
in $G_{i,t}$, 
and sets $\widetilde{P}_i$ as the shortest path from $r_i$ to the set of MV servers in $G_{i,t}$. 
Due to the construction of $G_{i,t}$, the output $P_i$ (respectively, $\widetilde{P}_i$) 
is a real (respectively, virtual) 
$M^{OFF}$-augmenting path from $r_i$ to some free real (respectively, virtual) servers. 
Moreover, every real or virtual $M^{OFF}$-augmenting path starting from $r_i$ 
has a corresponding path in $G_{i,t}$. Thus, it suffices to show that 
for any path $P$ in $G_{i,t}$, its distance in $G_{i,t}$ 
equals $\varphi(P)$ (or $\varphi_{t}(P)$ if $P$ connects $r_i$ to an 
MV server).  
For any path $P$ in $G_{i,t}$, define 
$sl(P)$ as the total edge weight in $P$.
The correctness of the subroutine is thus a corollary of the following lemmas.

\begin{lemma}\label{lemma: phisl}
Let $P$ be any path from $r_i$ to any free real server in $G_{i,t}$. 
Then $\varphi(P) = sl(P)$.
\end{lemma}

\begin{proof}
By the construction of $G_{i,t}$, 
we can write $P$ as $r'_1s'_1r'_2s'_2 \cdots r'_{\ell}s'_{\ell}$, where $r'_1 = r_i$ and $s'_{\ell}$ is a 
real free server. 
Moreover, for all $1 \leq k \leq \ell-1$, $\overrightarrow{s'_k, r'_{k+1}} \in M^{OFF}$. 
Thus, 
\begin{align*}
\varphi(P) 
= &\sum_{k=1}^{\ell}{\gamma D(r'_k,s'_k)}     
                   -\sum_{k=1}^{\ell -1}{D(s'_k, r'_{k+1})}\\   
\stackrel{Invariant~\eqref{eq: dualtight}}{=}        
  &\sum_{k=1}^{\ell}{\gamma D(r'_k,s'_{k})}  
                   -\sum_{k=1}^{\ell - 1}{\big(z(s'_k)+z(r'_{k+1})\big)}\\         
= &\sum_{k=1}^{\ell}{\left(sl(\overrightarrow{r'_k,s'_{k}})+z(r'_k) + z(s'_{k})\right)}  
                   -\sum_{k=1}^{\ell - 1}{\big(z(s'_k)+z(r'_{k+1})\big)}\\                     
= &\sum_{k=1}^{\ell}{sl(\overrightarrow{r'_k,s'_{k}})} + z(r'_1) + z(s'_{\ell})\\
\stackrel{Invariant~\eqref{eq: dual0}}{=} 
   &\sum_{k=1}^{\ell}{sl(\overrightarrow{r'_k,s'_{k}})}\\
= &sl(P), 
\end{align*}
where the last equality holds because by the construction of $G_{i,t}$, 
for all $1 \leq k \leq \ell -1$, $sl(\overrightarrow{s'_k,r'_{k+1}}) = 0$.  
\end{proof}

\begin{lemma}\label{lemma: phisl2}
Let $\widetilde{P}$ be any path from $r_i$ to any MV server in $G_{i,t}$. 
Then $\varphi_{t}(\widetilde{P}) = sl(\widetilde{P})$.
\end{lemma}

\begin{proof}
By the construction of $G_{i,t}$, 
we can write $\widetilde{P}$ as $r'_1s'_1r'_2s'_2 \cdots r'_{\ell}s'_{\ell}$, where $r'_1 = r_i$ 
and $s'_{\ell}$ is the MV server of some request $r_p$. Thus, $s'_{\ell} = \widetilde{s}_p$ and 
$r'_{\ell} = r_p$. 
Moreover, for all $1 \leq k \leq \ell-1$, $\overrightarrow{s'_k, r'_{k+1}} \in M^{OFF}$. 
Thus, 
\begin{align*}
\varphi_{t}(\widetilde{P}) 
= &\gamma D_t(r_p, \widetilde{s}_p) + \sum_{k=1}^{\ell-1}{\gamma D(r'_k,s'_{k})}     
                   -\sum_{k=1}^{\ell-1}{D(s'_k, r'_{k+1})}\\   
\stackrel{Invariant~\eqref{eq: dualtight}}{=}        
   &\gamma D_t(r_p, \widetilde{s}_p) + \sum_{k=1}^{\ell-1}{\gamma D(r'_k,s'_{k})}    
                   -\sum_{k=1}^{\ell-1}{\big(z(s'_k)+z(r'_{k+1})\big)}\\         
= &\sum_{k=1}^{\ell}{\left(sl(\overrightarrow{r'_k,s'_{k}})+z(r'_k) + z(s'_k)\right)}  
                   -\sum_{k=1}^{\ell-1}{\big(z(s'_k)+z(r'_{k+1})\big)}\\                     
= &\sum_{k=1}^{\ell}{sl(\overrightarrow{r'_k,s'_{k}})} + z(r'_1) + z(s'_{\ell})\\
\stackrel{Invariant~\eqref{eq: dual0}}{=} 
   &\sum_{k=1}^{\ell}{sl(\overrightarrow{r'_k,s'_{k}})}\\
= &sl(\widetilde{P}), 
\end{align*}
where the last equality holds because by the construction of $G_{i,t}$, 
for all $1 \leq k \leq \ell-1$, $sl(\overrightarrow{s'_k,r'_{k+1}}) = 0$.  
\end{proof}

\subsection{Proof of the Invariants}\label{appendix: dual}
\begin{lemma}
Assume that Invariants~\eqref{eq: dualrelaxed}, \eqref{eq: dualrelaxed2}, and \eqref{eq: dual0} hold 
before Step~\ref{step1} is executed. 
Then after Step~\ref{step1}, these invariants still hold.
\end{lemma}

\begin{proof}
Assume that the algorithm augments the offline matching by $P_i$ at time $t$. 
In the following proof, $M^{pre}$ and $M^{post}$ refer to the offline matching 
right before and right after it is augmented by $P_i$, respectively.
Moreover, $z^{pre}(\cdot)$ and $z_1(\cdot)$ refer to the dual variables 
right before and right after Step~\ref{step1} is executed, respectively. 
Thus, by the assumption of the lemma, Invariants~\eqref{eq: dualrelaxed}, \eqref{eq: dualrelaxed2}, 
and \eqref{eq: dual0} hold with respect to 
$z^{pre}$ and $M^{pre}$. 
Let $G^{pre} = G_{i, t}$ be the slack graph of $r_i$ right before the offline matching is augmented by $P_i$. 
Thus, $G^{pre}$ is constructed with respect to $M^{pre}$ and $z^{pre}$. 

\paragraph{Proof of Invariant~\eqref{eq: dual0}.} 
Let $r$ be any request that is not saturated by $M^{post}$. 
Thus, $r$ is not saturated by $M^{pre}$ and $r \neq r_i$. Therefore, 
$r$ is not in $G^{pre}$, and Step~\ref{step1} does not change $z(r)$. 
We then have $z_1(r) = z^{pre}(r) = 0$ as desired.
Let $s$ be any real server that is not saturated by $M^{post}$. 
If $s$ arrives after $r_i$ is matched, then $s$ is not in $G^{pre}$ 
and thus Step~\ref{step1} does not change $z(s)$. 
Thus, $z_1(s) = z^{pre}(s) = 0$ as desired. 
If $s$ arrives by the time $r_i$ is matched, $s$ is in $G^{pre}$. 
Recall that $P_i$ is the shortest path from $r_i$ to the set of free real servers in $G^{pre}$. 
The fact that $s$ is not saturated by $M^{post}$ implies that $sl(r_i, s) \geq sl(r_i, s^*)$ (recall that 
$s^* = \ter(P_i)$). 
Thus, after Step~\ref{step1}, $z_1(s) = z^{pre}(s) = 0$. 
Finally, because Step~\ref{step1} never changes the dual variable of a MV server, 
we have $z_1(\widetilde{s}) = z^{pre}(\widetilde{s}) = 0$ for any MV server $\widetilde{s}$.

\paragraph{Proof of Invariant~\eqref{eq: dualrelaxed2}.} 
Let $r_p \in R$ be a request such that $t \geq a(r_p)$. 
Note that the RHS of Invariant~\eqref{eq: dualrelaxed2} increases as $t$ increases. 
Thus, it suffices to prove
$z_1(r_p)+z_1(\widetilde{s}_p) \leq \gamma(t-a(r_p))$. 
If $sl(r_i, r_p) \geq sl(r_i, s^*)$, then Step~\ref{step1} does not change $z(r_p)$. 
Thus, by Invariant~\eqref{eq: dual0}, $z_1(r_p)+z_1(\widetilde{s}_p) = z^{pre}(r_p)+z^{pre}(\widetilde{s}_p) \leq \gamma(t-a(r_p))$ as desired. 

Next, we consider the case where $sl(r_i, r_p) < sl(r_i, s^*)$. 
In this case, we increase $z(r_p)$ by $sl(r_i, s^*) - sl(r_i, r_p)$. 
We only need to show that the increase of $z(r_p)$ 
is at most the slack between $r_p$ and $\widetilde{s}_p$. 
In other words, it suffices to show $sl(r_i, s^*) - sl(r_i, r_p) 
\leq sl(\overrightarrow{r_p, \widetilde{s}_p})$, 
or equivalently, $sl(r_i, s^*) \leq sl(r_i, r_p) 
                 +sl(\overrightarrow{r_p, \widetilde{s}_p})$. 
Observe that by the construction of $G^{pre}$, the RHS of this inequality is $sl(r_i, \widetilde{s}_p)$.

Let $\widetilde{s} = \ter(\widetilde{P}_i)$. Because $r_i$ is ready at time $t$, we have 
$sl(r_i, \widetilde{s}) = \varphi_{t}(\widetilde{P}_i) \geq \varphi(P_i) = sl(r_i, s^*)$. 
Because $\widetilde{P}_i$ is $r_i$'s virtual minimum augmenting path at time $t$, 
we have $sl(r_i, \widetilde{s}_p) \geq sl(r_i, \widetilde{s})$. 
Combining the above two inequalities, we then have 
$sl(r_i, \widetilde{s}_p) \geq sl(r_i, s^*)$ as desired.

\paragraph{Proof of Invariant~\eqref{eq: dualrelaxed}.} 
For any $r \in R$ and $s \in S$, we will prove 
$z_1(r) + z_1(s) \leq \gamma D(r,s)$.
Let $R'$ be the set of requests $r$ such that $r$ is not in $G^{pre}$ or $sl(r_i, r) \geq sl(r_i, s^*)$. 
Thus, if a request $r$ is not in $R'$, then $r$ is in $G^{pre}$ and $sl(r_i, r) < sl(r_i, s^*)$.
We divide the proof into cases. 
\begin{description}
\item[Case 1: $r \in R'$.] 
In this case, Step~\ref{step1} does not change $z(r)$ and thus $z_1(r) = z^{pre}(r)$. 
Because Step~\ref{step1} cannot increase the dual variable of any server, we have $z_1(s) \leq z^{pre}(s)$.
Thus, $z_1(r) + z_1(s) \leq z^{pre}(r) + z^{pre}(s) \leq \gamma D(r,s)$. 

\item[Case 2: $r \notin R'$ and $(r,s) \in M^{pre}$.] 
In this case, there is a directed edge $\overrightarrow{s,r}$ in $G^{pre}$, 
and the only way to reach $r$ from $r_i$ in $G^{pre}$ is via $s$. 
Because $sl(\overrightarrow{s,r}) = 0$, we have $sl(r_i, s) = sl(r_i, r) < sl(r_i, s^*)$. 
Thus, after Step~\ref{step1}, 
$z_1(r) + z_1(s) = z^{pre}(r) + (sl(r_i, s^*) - sl(r_i, r)) + z^{pre}(s) - (sl(r_i, s^*) - sl(r_i, s)) 
= z^{pre}(r) + z^{pre}(s) \leq \gamma D(r,s)$. 

\item[Case 3: $r \notin R'$, $(r,s) \notin M^{pre}$, and $s$ is in $G^{pre}$.]
In this case, there is a directed edge $\overrightarrow{r,s}$ in $G^{pre}$.  
Thus, $sl(r_i, s) \leq sl(r_i, r) + sl(\overrightarrow{r,s})$, or equivalently, 
$sl(r_i, s) - sl(r_i, r) \leq sl(\overrightarrow{r,s})$.
\begin{description}
\item[Case 3A: $sl(r_i, s) < sl(r_i, s^*)$.] 
In this subcase, after Step~\ref{step1}, we have
\begin{align*}
&z_1(s) + z_1(r) \\
=    &z^{pre}(s) - (sl(r_i, s^*) - sl(r_i, s)) + z^{pre}(r) + (sl(r_i, s^*) - sl(r_i, r))\\
\leq &z^{pre}(s) + z^{pre}(r) + sl(\overrightarrow{r,s})\\
=    &z^{pre}(s) + z^{pre}(r) + \gamma D(r,s) - (z^{pre}(r)+z^{pre}(s))\\
=    &\gamma D(r,s). 
\end{align*}
\clearpage
\item[Case 3B: $sl(r_i, s^*) \leq sl(r_i, s)$.] 
Because $sl(r_i, s) \leq sl(r_i, r) + sl(\overrightarrow{r,s})$, 
in this subcase, we have 
$sl(r_i, s^*) \leq sl(r_i, r) + sl(\overrightarrow{r,s})$, 
or equivalently, $sl(r_i, s^*) - sl(r_i, r) \leq sl(\overrightarrow{r,s})$. 
Thus, after Step~\ref{step1},
\begin{align*}
&  z_1(s) + z_1(r) \\
=    &z^{pre}(s) + z^{pre}(r) + (sl(r_i, s^*) - sl(r_i, r))\\
\leq &z^{pre}(s) + z^{pre}(r) + sl(\overrightarrow{r,s})\\
=    &z^{pre}(s) + z^{pre}(r) + \gamma D(r,s) - (z^{pre}(r)+z^{pre}(s))\\
=    &\gamma D(r,s). 
\end{align*}
\end{description}
\item[Case 4: $r \notin R'$, $(r,s) \notin M^{pre}$, and $s$ is not in $G^{pre}$.] 
In this case, because $z_1(s) = 0$ (by Invariant~\eqref{eq: dual0}), 
it suffices to show $z_1(r) \leq \gamma D(r,s)$. 
Observe that because $s$ is not in $G^{pre}$, $a(s) \geq t$. 
Thus, $D(r,s) \geq a(s) - a(r) \geq t - a(r)$. 
Let $\widetilde{s}$ be the MV server of $r$. 
By Invariant~\eqref{eq: dualrelaxed2}, $z_1(r)+z_1(\widetilde{s}) \leq \gamma(t - a(r))$. 
Because $z_1(\widetilde{s}) = 0$, we then have $z_1(r) \leq \gamma(t - a(r)) \leq \gamma D(r,s)$.
\end{description}
\end{proof}

Because Step~\ref{step2} only decreases dual variables for requests in $P_i$, 
Invariants~\eqref{eq: dualrelaxed}, 
\eqref{eq: dualrelaxed2}, and \eqref{eq: dual0} still hold 
after Step~\ref{step2}. 
Thus, combining the above lemma, we have the following lemma.
\begin{lemma} \label{lemma: invariants}
Assume that Invariants~\eqref{eq: dualrelaxed}, \eqref{eq: dualrelaxed2}, \eqref{eq: dual0} hold 
before Step~\ref{step1} is executed. 
Then after Step~\ref{step2}, these invariants still hold.
\end{lemma}

\begin{lemma} \label{lemma: dualtight}
Assume that Invariant~\eqref{eq: dualtight} holds before Step~\ref{step1} is executed. 
Then after Step~\ref{step2},  Invariant~\eqref{eq: dualtight} still holds.
\end{lemma}

\begin{proof}
Assume that the algorithm augments the offline matching by $P_i$ at time $t$. 
In the following proof, $M^{pre}$ and $M^{post}$ refer to the offline matching 
right before and right after it is augmented by $P_i$, respectively.
Moreover, $z^{pre}(\cdot)$, $z_1(\cdot)$, and $z_2(\cdot)$ 
refer to the dual variables right before Step~\ref{step1} is executed, 
the dual variables right after Step~\ref{step1} is executed, 
and the dual variables right after Step~\ref{step2} is executed, 
respectively.
It suffices to prove $z_2(r) + z_2(s) = D(r,s)$ holds 
for any $(r,s) \in M^{post}$. 

We first consider the case where $(r,s) \notin E(P_i)$. 
In this case, $(r,s) \in M^{pre}$ and thus $z^{pre}(r) + z^{pre}(s) = D(r,s)$. 
Because $(r,s) \in M^{pre}$, there is a directed edge $\overrightarrow{s,r}$ in $G^{pre}$, 
and the only way to reach $r$ from $r_i$ in $G^{pre}$ is via $s$. 
In addition, $sl(\overrightarrow{s,r}) = 0$. Thus, $sl(r_i, r) = sl(r_i, s)$. 
After Step~\ref{step1}, we then have 
$z_1(r) + z_1(s) = z^{pre}(r) + z^{pre}(s) = D(r,s)$. 
Because $(r, s) \in M^{post}$ and $(r,s) \notin P_i$, 
$r$ is not in $P_i$.
Thus, Step~\ref{step2} does not change $z(r)$, and we have  
$z_2(r) + z_2(s) = z_1(r) + z_1(s) = D(r,s)$. 

Next, we consider the case where $(r,s) \in P_i$. 
Because $(r, s) \in M^{post}$, there is a
directed edge $\overrightarrow{r,s}$ in $\overrightarrow{E}(P_i)$. 
Recall that $P_i$ is the shortest path from $r_i$ to the set of  
real free servers in $G^{pre}$. 
The property of shortest path thus implies that 
\[
sl(r_i, s) = sl(r_i, r) + sl(\overrightarrow{r, s})
\]
and 
\[
sl(r_i, r) \leq sl(r_i, s^*) \text{ and } sl(r_i, s) \leq sl(r_i, s^*).
\]
Thus, after Step~\ref{step1}, we have 
\begin{align*}
  &z_1(s) + z_1(r) \\
= &z^{pre}(s) - (sl(r_i, s^*) - sl(r_i, s)) 
                 + z^{pre}(r) + (sl(r_i, s^*) - sl(r_i, r))\\
=&z^{pre}(s) + z^{pre}(r) + sl(r_i, s) - sl(r_i, r)\\
=&z^{pre}(s) + z^{pre}(r) + sl(\overrightarrow{r, s})\\
=&\gamma D(r,s)
\end{align*}
Because $\overrightarrow{r,s}$ is in $\overrightarrow{E}(P_i)$,
after Step~\ref{step2}, we then have 
\[
z_2(r) + z_2(s) = z_1(r) - (\gamma -1)D(r,s) + z_1(s) = D(r,s).
\]
\end{proof}

Because all invariants hold initially, 
by Lemmas~\ref{lemma: invariants} and \ref{lemma: dualtight} and induction, 
all invariants hold throughout the execution of the VRM algorithm.

\subsection{Proof of Lemma~\ref{lemma: DUB}}\label{appendix: lemma: DUB}
For any path $P$, define $D(P) = \sum_{(u,v)\in E(P)}{D(u,v)}$ as $P$'s distance in the TA plane.
By the triangle inequality, we have 
\[
D(M^{VRM}) = \sum_{i=1}^m{D(\ori(P^*_i), \ter(P^*_i))}
        \leq \sum_{i=1}^m{D(P^*_i)}. 
\]
Next, we argue that 
\begin{equation}\label{eq: DPiub}
\sum_{i=1}^m{\varphi(P^*_i)} \geq \frac{\gamma-1}{2}\sum_{i=1}^m{D(P^*_i)}, 
\end{equation}
which completes the proof. 

Fix any edge $(r, s)$. 
The first time (and every subsequent odd-numbered time) 
$(r, s)$ appears in some request $r_i$'s minimum augmenting path $P^*_i$ 
(i.e., $(r,s) \in E(P^*_i)$), 
$(r,s)$ contributes $\gamma D(r,s)$ to the LHS of Eq.~\eqref{eq: DPiub} 
(because $(r,s)$ is added to $M^{OFF}$). 
The second time (and every subsequent even-numbered time) $(r, s)$ appears 
in some request $r_i$'s minimum augmenting path $P^*_i$, 
$(r,s)$ contributes $-D(r,s)$ to the LHS of Eq.~\eqref{eq: DPiub} 
(because $(r,s)$ is removed from $M^{OFF}$). 
Thus, $(r,s)$'s amortized contribution to the LHS of Eq.~\eqref{eq: DPiub} 
is at least $\frac{\gamma D(r,s) -D(r,s))}{2} = \frac{\gamma-1}{2}D(r,s)$. 
Specifically, we have 
\[
     \sum_{i=1}^m{\varphi(P^*_i)} 
\geq \sum_{i=1}^m{\sum_{(r,s) \in E(P^*_i)}{\frac{\gamma-1}{2}D(r,s)}}
=    \frac{\gamma-1}{2}\sum_{i=1}^m{D(P^*_i)}.
\]

\subsection{Proofs of the Nonnegativity of $\gamma$-Net-Costs}\label{appendix: properties}
\subsubsection{Proof of Lemma~\ref{lemma: realgeq0}}
Write $P$ as $r'_1s'_1r'_2s'_2\cdots r'_{\ell}s'_{\ell}$, 
where $r'_1 \in R$ and $s'_{\ell} \in S$ are not saturated by $M^{OFF}$.
By Invariant~\eqref{eq: dual0}, we have 
\begin{equation}\label{eq: rs0}
z(r'_1) = z(s'_{\ell}) = 0.
\end{equation}
Thus, 
\begin{align*}
\varphi(P) 
= &\sum_{k=1}^{\ell}{\gamma D(r'_k,s'_{k})}     
                   -\sum_{k=1}^{\ell-1}{D(s'_k, r'_{k+1})}\\   
\stackrel{Invariant~\eqref{eq: dualtight}}{=}        
  &\sum_{k=1}^{\ell}{\gamma D(r'_k,s'_{k})}  
                   -\sum_{k=1}^{\ell-1}{\big(z(s'_k)+z(r'_{k+1})\big)}\\         
\stackrel{Eq.~\eqref{eq: rs0}}{=}        
  &\sum_{k=1}^{\ell}{\gamma D(r'_k,s'_{k})}  
                   - \sum_{k=1}^{\ell-1}{\big(z(s'_k)+z(r'_{k+1})\big)} 
                   - z(r'_1)                   
                   - z(s'_{\ell})\\       
=&\sum_{k=1}^{\ell}{\big(\gamma D(r'_k,s'_k) - (z(r'_k)+z(s'_{k}))\big)} \\                     
\stackrel{Invariant~\eqref{eq: dualrelaxed}}{\geq} &0. 
\end{align*}

\subsubsection{Proof of Lemma~\ref{lemma: virtualgeq0}}
Write $\widetilde{P}$ as 
$r'_1s'_1r'_2s'_2\cdots r'_{\ell}, \widetilde{s}_p$, where $r'_1 \in R$ is not 
saturated by $M^{OFF}$ and $t \geq a(\widetilde{r}_p)$.
By Invariant~\eqref{eq: dual0}, we then have 
\begin{equation}\label{eq: rstilde0}
z(r'_1) = z(\widetilde{s}_p) = 0.
\end{equation}
Moreover, by the definition of virtual augmenting path, we have
\begin{equation}\label{eq: vaug}
r'_{\ell} = r_p.
\end{equation}

Thus, 
\begin{align*}
\varphi_{t}(\widetilde{P}) 
= &\gamma D_t(r_p, \widetilde{s}_p) + \sum_{k=1}^{\ell-1}{\gamma D(r'_k,s'_{k})}     
                   -\sum_{k=1}^{\ell-1}{D(s'_k, r'_{k+1})}\\   
\stackrel{Invariant~\eqref{eq: dualtight}}{=}        
  &\gamma D_t(r_p, \widetilde{s}_p) + \sum_{k=1}^{\ell-1}{\gamma D(r'_k,s'_{k})}    
                   -\sum_{k=1}^{\ell-1}{\big(z(s'_k)+z(r'_{k+1})\big)}\\         
\stackrel{Eq.~\eqref{eq: rstilde0}}{=}        
  &\gamma D_t(r_p, \widetilde{s}_p) + \sum_{k=1}^{\ell-1}{\gamma D(r'_k,s'_{k})}    
                   -\sum_{k=1}^{\ell-1}{\big(z(s'_k)+z(r'_{k+1})\big)}
                   -z(r'_1)                   
                   -z(\widetilde{s}_p)\\       
\stackrel{Eq.~\eqref{eq: vaug}}{=} 
  &\gamma D_t(r_p, \widetilde{s}_p) -z(r_p) -z(\widetilde{s}_p)
  + \sum_{k=1}^{\ell-1}{\big( \gamma D(r'_k,s'_{k}) 
                            -(z(r'_k)+z(s'_{k})) \big)}  \\  
\stackrel{Invariants~\eqref{eq: dualrelaxed2}\&\eqref{eq: dualrelaxed}}{\geq} &0. 
\end{align*}

\clearpage
\section{Proof of Eq.~\eqref{eq: maingoal1}}\label{appendix: eq: maingoal1}
In the following proof, define $\Phi_i$ as $\varphi(P^*_i)$.

\subsection{Partition of the Request Set}
We first partition the request set $R$ into $O(\log m)$ groups 
$R_0, R_1, R_2, \cdots$ based on $\Phi_i$. Specifically, 
$R_0$ consists of all requests $r_i$ such that 
$\Phi_i < \frac{D(M^{OPT})}{m}$. 
For $\ell \geq 1$, $R_{\ell}$ consists of all requests $r_i$ such that 
$\frac{2^{\ell-1}D(M^{OPT})}{m} \leq \Phi_i < \frac{2^{\ell}D(M^{OPT})}{m}$. 
We use the following lemma to upper bound the number of groups.
\begin{lemma}\label{lemma: phiUB}
For any $r_i \in R$, 
$\Phi_i \leq \gamma \cdot D(M^{OPT})$.
\end{lemma}
\begin{proof}
Let $M^{OFF}$ be the offline matching right before it is augmented by $P^*_i$. 
Let $t$ be the time when $r_i$ is matched by the VRM algorithm.
Consider $G' = (S \cup R, M^{OFF} \oplus M^{OPT})$. 
In $G'$, there is an $M^{OFF}$-augmenting path $P$ that originates at $r_i$. 
We first consider the case where $\ter(P)$ is a server that already arrives at the system.  
By the optimality of $P^*_i$, 
we then have $\Phi_i = \varphi(P^*_i) \leq \varphi(P)$. 
Because $P$ alternates between edges in $M^{OFF}$ and $M^{OPT}$, 
we then have 
\begin{align}
\nonumber \varphi(P) &= \gamma \sum_{(r, s) \in E(P) \cap M^{OPT}}{D(r, s)} 
- \sum_{(r, s) \in E(P) \cap M^{OFF}}{D(r, s)} \\
\label{eq: ana42}&\leq \gamma \sum_{(r, s) \in E(P) \cap M^{OPT}}{D(r, s)}
\leq  \gamma \cdot D(M^{OPT}).
\end{align}

Next, we consider the case where $\ter(P)$ arrives at the system 
after $r_i$ is matched (and thus $a(\ter(P)) \geq t$). 
Let $r_p$ be the last request in $P$ (and thus $\overrightarrow{r_p, \ter(P)}$ 
is the last directed edge in $P$). Construct a virtual augmenting path 
$\widetilde{P}$ by 
replacing $\ter(P)$ with $\widetilde{s}_p$. 
Thus, $D(r_p, \ter(P)) \geq a(\ter(P)) - a(r_p) \geq t-a(r_p) 
= D_t(r_p, \tilde{s}_p)$, which implies 
$\varphi(P) \geq \varphi_t(\widetilde{P})$. 
By the design of the VRM algorithm, when $r_i$ becomes ready and matched at time 
$t$, it must be the case that 
$\varphi(P^*_i) \leq \varphi_t(\widetilde{P}_i)$.
Thus, we have 
$\varphi(P^*_i) \leq \varphi_t(\widetilde{P}_i) \leq \varphi_t(\widetilde{P})$.  
Because $\varphi_t(\widetilde{P}) \leq \varphi(P)$, we then have
$\Phi_i = \varphi(P^*_i) \leq \varphi(P)$. 
The proof then follows from Eq.~\eqref{eq: ana42}.
\end{proof}

By Lemma~\ref{lemma: phiUB} and $\gamma = O(1)$, 
the number of groups is $O(\log m)$.
Observe that $\sum_{r_i \in R_0} \Phi_i < D(M^{OPT})$. 
In the remainder of the proof, we fix a request group $R_{\ell}$ 
with $\ell \geq 1$. 
Let $\Phi = \frac{2^{\ell-1}D(M^{OPT})}{m}$ and thus for any $r_i \in R_{\ell}$
\begin{equation}\label{eq: phiLUB}
\Phi \leq \Phi_i < 2\Phi.
\end{equation} 
To prove Eq.~\eqref{eq: maingoal1}, it is then sufficient to prove 
\begin{equation}\label{eq: maingoal2}
\sum_{r_i \in R_{\ell}} \Phi = O(\sqrt{m} \log m)D(M^{OPT}).
\end{equation}

\begin{remark}
In the partition in~\cite{nayyar2017input}, there is an additional group consisting of 
requests $r_i$ such that $\Phi_i \leq 16\gamma D(r_i, opt(r_i))$, where $opt(r_i)$ is the 
server that matches to $r_i$ in $M^{OPT}$. This group is used to facilitate 
input sensitive analysis, and thus is omitted in our proof.  
\end{remark}

We then partition $R_{\ell}$ into $K$ clusters $C_1, C_2, \cdots, C_K$ 
such that each cluster $C_k$ has a \textbf{center request} 
$c_k \in C_k$ satisfying the following properties:
\begin{description}
\item[P1:] For any two distinct clusters $C_k$ and $C_{k'}$, 
$D(c_k, c_{k'}) \geq \frac{\Phi}{2\gamma}$.
\item[P2:] For any cluster $C_k$ and any request $r$ in $C_k$, 
$D(c_k, r) < \frac{\Phi}{2\gamma}$.
\item[P3:] For any cluster $C_k$, the center request $c_k$ is the first matched 
request in $C_k$ under the VRM algorithm.
\end{description}
These clusters can be constructed as follows:
Whenever a request $r$ is matched by the VRM algorithm, 
we find the center request $c_k$ that is closest to $r$. 
If $D(r, c_k) < \frac{\Phi}{2\gamma}$, we add $r$ to $C_k$. 
Otherwise, we create a new cluster and set $r$ as its center request. 

\subsection{Optimal Paths}
For each center request $c_k$, let $M_k$ be the offline matching right before 
$c_k$ is matched. Observe that by P3, all requests in $C_k$ are not saturated by 
$M_k$. Thus, for each request $r_i \in C_k$, $M^{OPT} \oplus M_k$ has 
an $M_k$-augmenting path, denoted by $P^{OPT}_i$, that originates at $r_i$. 
$P^{OPT}_i$ is called the \textbf{optimal path} for $r_i$. We have the following 
results. Recall that for any path $P$, 
$D(P) = \sum_{(u,v) \in E(P)}{D(u,v)}$. 
Moreover, for any matching $M$, define $M \cap P = M \cap E(P)$, 
$M \setminus P = M \setminus E(P)$, and $P \setminus M = E(P) \setminus M$.
\begin{fact}\label{fact: vdisjoint}
For any two requests in the same cluster, 
these requests' optimal paths are vertex-disjoint.
\end{fact} 

\begin{lemma}\label{lemma: OPTLB1}
Let $C_k$ be any cluster and $r_i \in C_k$. 
Then $D(M^{OPT} \cap P^{OPT}_i) \geq \frac{1}{\gamma+1}D(P^{OPT}_i)$.
\end{lemma}
\begin{proof}
Observe that $P^{OPT}_i$ is a real $M_k$-augmenting path. 
By Lemma~\ref{lemma: realgeq0}, $\varphi(P^{OPT}_i) \geq 0$ and thus
\begin{align*}
\nonumber \gamma \cdot D(M^{OPT} \cap P^{OPT}_i) 
- D(P^{OPT}_i \setminus M^{OPT}) &\geq 0\\
\nonumber
(\gamma+1) \cdot D(M^{OPT} \cap P^{OPT}_i) 
&\geq D(M^{OPT} \cap P^{OPT}_i) + D(P^{OPT}_i \setminus M^{OPT})\\ 
(\gamma+1) \cdot D(M^{OPT} \cap P^{OPT}_i) &\geq D(P^{OPT}_i).
\end{align*}
\end{proof}

For every request $r_i \in C_k$, define 
$\delta_i = \max_{v \in P^{OPT}_i}D(c_k, v)$. 
Next, we give upper and lower bounds of $\delta_i$. 
\begin{lemma}\label{lemma: deltaLB}
Let $C_k$ be any cluster and $r_i \in C_k$. 
Then $\delta_i \geq \frac{\Phi}{\gamma}$.
\end{lemma}
\begin{proof}
Observe that $P^{OPT}_i$ must contain a server that is not 
saturated by $M_k$. 
Let $s$ be any such a server. 
It suffices to prove $D(c_k, s) \geq \frac{\Phi}{\gamma}$. 
We first consider the case where $s$ arrives before $c_k$ is matched 
by the VRM algorithm. 
In this case, the path consisting of only $\overrightarrow{c_k, s}$ 
is a valid $M_k$-augmenting path of $c_k$. 
By the optimality of $c_k$'s minimum augmenting path and Eq.~\eqref{eq: phiLUB}, 
we have $\Phi \leq \gamma \cdot D(c_k, s)$ as desired. 

Next, we consider the case where $s$ arrives after $c_k$ is matched by the 
VRM algorithm. Let $t$ be the time when $c_k$ is matched by the VRM algorithm. 
Let $\tilde{s}$ be the MV server of $c_k$. 
Thus, $D(c_k, s) \geq D_t(c_k, \tilde{s})$. 
By the design of the VRM algorithm, at time $t$, 
$c_k$ becomes ready, which combined with Eq.~\eqref{eq: phiLUB} 
implies $\Phi \leq \gamma \cdot D_t(c_k, \tilde{s}) 
\leq \gamma \cdot D(c_k, s)$ 
as desired. 
\end{proof}

To upper bound $\delta_i$, we need the following property about $M^{OFF}$,  
which holds throughout the execution of the VRM algorithm. 

\begin{lemma}\label{lemma: diffM}
Let $M$ be any matching that saturates all the vertices in $M^{OFF}$. 
Then $D(M^{OFF}) \leq \gamma D(M)$.
\end{lemma}
\begin{proof}
By Invariant~\eqref{eq: dualrelaxed}, we have 
\[
\gamma \cdot D(M) = \sum_{(r,s)\in M}{\gamma \cdot D(r,s)} 
               \geq \sum_{(r,s)\in M}{(z(r)+z(s))} 
                  = \sum_{(r,s)\in M^{OFF}}{(z(r)+z(s))},
\]
where the last inequality is due to Invariant~\eqref{eq: dual0} and 
the assumption that $M$ saturates all the vertices in $M^{OFF}$. 
By Invariant~\eqref{eq: dualtight}, 
we have
\[
\sum_{(r,s)\in M^{OFF}}{(z(r)+z(s))} = \sum_{(r,s)\in M^{OFF}}{D(r,s)} = D(M^{OFF}),
\]
which then completes the proof.
\end{proof}

\begin{lemma}\label{lemma: deltaUB}
Let $C_k$ be any cluster and $r_i \in C_k$. 
Then $\delta_i < (\gamma+\frac{3}{2})D(M^{OPT})$.
\end{lemma}

\begin{proof}
By the triangle inequality, 
$\delta_i \leq D(c_k,r_i) + D(P^{OPT}_i)  
< \frac{\Phi}{2\gamma} + D(P^{OPT}_i)$, 
where the last inequality is due to P2. 
By Eq.~\eqref{eq: phiLUB} and Lemma~\ref{lemma: phiUB}, 
$\Phi \leq \gamma \cdot D(M^{OPT})$ 
and thus
$\frac{\Phi}{2\gamma} \leq \frac{D(M^{OPT})}{2}$.
On the other hand, we have $D(P^{OPT}_i) \leq D(M_k) + D(M^{OPT}) 
\leq (\gamma+1)D(M^{OPT})$, where the last inequality is due to 
Lemma~\ref{lemma: diffM}. 
As a result, we have 
$\delta_i < (\gamma + \frac{3}{2})D(M^{OPT})$.
\end{proof}

\subsection{Outer Groups and Balls}
We further partition $R_{\ell}$ into \textit{outer groups}
according to $\delta_i$. 
Specifically, $r_i \in R_{\ell}$ is in outer group $R^{\lambda}$ if 
\begin{equation}\label{eq: lambda}
\frac{2^{\lambda-1}D(M^{OPT})}{\gamma m} \leq \delta_i 
< \frac{2^{\lambda}D(M^{OPT})}{\gamma m}. 
\end{equation}
By Lemma~\ref{lemma: deltaLB} and the definition of $\Phi$, 
$\lambda \geq \ell \geq 1$. 
By Lemma~\ref{lemma: deltaUB}, $\lambda = O(\log(\gamma m))$. 
As a result, the number of outer groups is 
$O(\log(\gamma m))$. Because $\gamma = O(1)$, the number of outer groups is 
$O(\log m)$. 

In the remainder of the proof, we fix a outer group $R^{\lambda}$. 
To prove Eq.~\eqref{eq: maingoal2}, 
it suffices to prove  
\begin{equation}\label{eq: maingoal3}
\sum_{r_i \in R^{\lambda}} \Phi = O(\sqrt{m})D(M^{OPT}).
\end{equation}

For each cluster $C_k$, let $C'_k = C_k \cap R^{\lambda}$. 
Moreover, define $B_k$ as a ball centered at $c_k$ with radius 
$\rho = \frac{2^{\lambda}D(M^{OPT})}{\gamma m}$. 
By Eq.~\eqref{eq: lambda}, $\frac{\rho}{2} \leq \delta_i < \rho$. 
We then have the following simple fact.
\begin{fact}\label{fact: outball}
For every $r_i \in C'_k$, 
$B_k$ contains all the vertices in $P^{OPT}_i$. 
\end{fact}

For any set of edges $M$, define $M \cap B_k$ as 
the set of edges in $M$ whose both endpoints are in $B_k$. 
\begin{lemma}\label{lemma: OPTLB2}
For any $k \in \{1,2, \cdots, K\}$, 
$D(M^{OPT} \cap B_k) > \frac{\rho|C'_k|}{4(\gamma+1)}$.
\end{lemma}
\begin{proof}
By Facts~\ref{fact: vdisjoint} and \ref{fact: outball}, 
it suffices to prove that for each $r_i \in C'_k$, 
$D(M^{OPT} \cap P^{OPT}_i) > \frac{\rho}{4(\gamma+1)}$. 
By the triangle inequality, $\delta_i \leq D(c_k, r_i) + D(P^{OPT}_i)$ or 
equivalently, $D(P^{OPT}_i) \geq \delta_i - D(c_k, r_i)$. 
By P2 and Lemma~\ref{lemma: deltaLB},
we have $D(c_k, r_i) < \frac{\Phi}{2\gamma} \leq \frac{\delta_i}{2}$. 
Thus, $D(P^{OPT}_i) > \delta_i - \frac{\delta_i}{2} 
= \frac{\delta_i}{2}$. 
Thus, by Lemma~\ref{lemma: OPTLB1}, 
$D(M^{OPT} \cap P^{OPT}_i) > \frac{1}{\gamma+1}\frac{\delta_i}{2}$. 
The proof then follows from $\delta_i \geq \frac{\rho}{2}$.
\end{proof}

\subsection{Proof of Eq.~\eqref{eq: maingoal3}}
We first state a variant of the Vitali's covering lemma proved in~\cite{nayyar2017input}. 
For any ball $B_k$, let $3B_k$ be the ball centered at $c_k$ with radius $3\rho$. 
Recall that we have $K$ clusters.
\begin{lemma}[\cite{nayyar2017input}]
There exists a set $H \subseteq \{1, 2, \cdots, K\}$ such that:
\begin{description}
\item[C1:] For any two distinct $h_1, h_2 \in H$, $B_{h_1} \cap B_{h_2} = \varnothing$.
\item[C2:] For any $k \in \{1,2, \cdots, K\}$, there exists $head(k) \in H$ satisfying the following properties:
    \begin{description}
        \item[C2A:] $B_k$ is contained in $3B_{head(k)}$.
        \item[C2B:] $|C'_k| \leq |C'_{head(k)}|$.
    \end{description}
\end{description}
\end{lemma} 

For any $h \in H$, define $cover(h)$ as the number of clusters that set head to $h$. 
Specifically, 
\[
cover(h) = |\{k| head(k) = h, 1 \leq k \leq K\}|.
\]
Thus, 
\[
  \sum_{r_i \in R^{\lambda}} \Phi 
= \sum_{k = 1}^K{|C'_k|\Phi} 
\stackrel{C2}{\leq}
  \sum_{h \in H}{|C'_h|cover(h)\Phi}.
\]
On the other hand, because $\gamma = O(1)$, we have 
\[
  D(M^{OPT}) 
\stackrel{C1}{\geq}
  \sum_{h \in H}{D(M^{OPT} \cap B_h)}  
\stackrel{Lemma~\ref{lemma: OPTLB2}}{>}
  \sum_{h \in H}{\frac{\rho|C'_h|}{4(\gamma+1)}}
=
  \Omega\left( \sum_{h \in H}{\rho|C'_h|} \right).
\]
Thus, to prove Eq.~\eqref{eq: maingoal3}, it suffices to prove $\frac{cover(h)\Phi}{\rho} = O(\sqrt{m})$ for 
any $h \in H$. Fix $h \in H$.  
Define $C$ as $\{c_k| head(k) = h, 1 \leq k \leq K\}$.  Thus, $|C| = cover(h)$. 
Let $\TSP(C)$ be the distance of the shortest cycle that visits every center request in $C$ in the TA plane.
Let $\Diam(C)$ be the distance of the farthest pair of center requests in $C$ in the TA plane.
By P1, for any two distinct $c_k, c_{k'} \in C$, 
we have $D(c_k, c_{k'}) \geq \frac{\Phi}{2\gamma}$. 
Thus, 
\[
\TSP(C) \geq |C|\frac{\Phi}{2\gamma} = \frac{cover(h)\Phi}{2\gamma}.\] 
By C2A, we have 
\[
\Diam(C)  \leq 6\rho.
\]
Combining the above two inequalities and $\gamma = O(1)$, we then have $\frac{cover(h)\Phi}{\rho} = 
O\left(\frac{\TSP(C)}{\Diam(C)}\right)$. 
Because TA plane is two-dimensional, 
$\frac{\TSP(C)}{\Diam(C)} = O\left(|C|^{1-\frac{1}{2}}\right) = O\left(\sqrt{m}\right)$~\cite{nayyar2017input},  
which completes the proof. 

\begin{remark} 
To see that $\frac{\TSP(C)}{\Diam(C)} = O\left(|C|^{1-\frac{1}{2}}\right)$, 
let $\Delta = \Diam(C)$. 
Let $A$ be the smallest axis aligned box that contains $C$ in the TA plane.  
Thus, $A$ has side lengths of $O(\Delta)$.
We can divide $A$ into $O(|C|)$ squares of side length $\frac{\Delta}{\sqrt{|C|}}$, 
and construct a cycle $W$ of length $O\left(|C| \cdot \frac{\Delta}{\sqrt{|C|}}\right)$ 
that visits all centers of these squares.  
For any $r \in C$, let $center(r)$ be the center of the square that contains $r$. 
For any $r \in C$, we then add to $W$ two directed edges $\overrightarrow{center(r), r}$  
and $\overrightarrow{r, center(r)}$ at a cost of $O\left(\frac{\Delta}{\sqrt{|C|}}\right)$. 
These edges and $W$ form a closed walk that visits every vertex in $C$, 
and the total distance is 
$O\left(|C| \cdot \frac{\Delta}{\sqrt{|C|}}\right)$. 
As a result, $\TSP(C) = O\left(|C| \cdot \frac{\Delta}{\sqrt{|C|}}\right)$, 
which implies $\frac{\TSP(C)}{\Diam(C)} = O\left(|C|^{1-\frac{1}{2}}\right)$.
\end{remark}

\clearpage
\section{Construction of $P'$ and $P''$}\label{appendix: lemma: AU}
Let $P$ and $\widetilde{P}$ be any real and virtual $M^{post}_{w+1}$-augmenting paths of $r_i$, 
respectively. In this section, our goal is to construct two real $M^{pre}_{w+1}$-augmenting paths of $r_i$, 
$P'$ and $P''$, such that
\begin{equation}\label{eq: P'goal}
\varphi(P') \leq \varphi(P)
\end{equation}
and 
\begin{equation}\label{eq: P''goal}
\varphi(P'') \leq \varphi_{t_{w+1}}(\widetilde{P}).
\end{equation} 

Recall that in $E_{w+1}$, the VRM algorithm augments $M^{pre}_{w+1}$
by request $r^*$'s minimum augmenting path $P^*$. 
To construct $P'$, we partition $P$ at the intersections of $P^*$ and $P$. 
The partition can be applied to any $r_i$'s $M^{post}_{w+1}$-augmenting path, be it real or virtual. 
Thus, in the following proof, we consider any $r_i$'s $M^{post}_{w+1}$-augmenting path $\bar{P}$ 
and discuss the partition of $\bar{P}$. 

\subsection{Partition of $\bar{P}$}
We call a directed edge $\overrightarrow{u,v}$ in $\bar{P}$ an \textbf{entry edge} 
of $\bar{P}$ if $(u,v)$ is not in $E(P^*)$ but $v$ is in $P^*$. 
On the other hand, a directed edge $\overrightarrow{u,v}$ in $\bar{P}$ 
is called an \textbf{exit edge} of $\bar{P}$ if $(u,v)$ is not in $E(P^*)$ 
but $u$ is in $P^*$. 

\begin{claim}\label{claim: entry}
Every entry edge of $\bar{P}$ goes from a request to a server. 
\end{claim}

\begin{proof}
Let $\overrightarrow{u, v}$ be an entry edge of $\bar{P}$. 
Because $v$ is in $P^*$ and any vertex in $P^*$ is saturated by $M^{post}_{w+1}$, 
$v$ is incident to some edge $e$ in $M^{post}_{w+1} \cap E(P^*)$. 
Thus, $(u, v)$, which is not in $E(P^*)$, cannot be in $M^{post}_{w+1}$ (otherwise, 
two distinct edges $e$ and $(u, v)$ in $M^{post}_{w+1}$ are incident to $v$). 
Because $\bar{P}$ is an $M^{post}_{w+1}$-augmenting path, 
by Fact~\ref{fact: s}, $u$ is a request and $v$ is a server.
\end{proof}

\begin{claim}\label{claim: exist}
Every exist edge of $\bar{P}$ goes from a request to a server. 
\end{claim}

\begin{proof}
Let $\overrightarrow{u, v}$ be an exit edge of $\bar{P}$. 
Because $u$ is in $P^*$ and any vertex in $P^*$ is saturated by $M^{post}_{w+1}$, 
$u$ is incident to some edge $e$ in $M^{post}_{w+1} \cap E(P^*)$. 
Thus, $(u, v)$, which is not in $E(P^*)$, cannot be in $M^{post}_{w+1}$. 
Because $\bar{P}$ is an $M^{post}_{w+1}$-augmenting path, 
by Fact~\ref{fact: s}, $u$ is a request and $v$ is a server.
\end{proof}

For any directed path $Q$ and any directed edge 
$\overrightarrow{u,v}$ in $\overrightarrow{E}(Q)$, define $next(Q, \overrightarrow{u,v})$ 
as the directed edge originating at $v$ in $\overrightarrow{E}(Q)$.

\begin{claim}\label{claim: entrynext}
Let $\overrightarrow{r, s}$ be any entry edge of $\bar{P}$. 
Let $\overrightarrow{s, r'} = next(\bar{P}, \overrightarrow{r, s})$. 
Then $\overrightarrow{r', s} \in \overrightarrow{E}(P^*)$.
\end{claim}

\begin{proof}
Because $\bar{P}$ is an $M^{post}_{w+1}$-augmenting path, by Fact~\ref{fact: s}, 
$(r, s)$ is not in $M^{post}_{w+1}$. 
Thus, $\bar{P}$'s next edge $(s, r')$ is in $M^{post}_{w+1}$. 
Because $\overrightarrow{r, s}$ is an entry edge, $s$ is in $P^*$. 
There is an edge $(s, r'')$ in $E(P^*) \cap M^{post}_{w+1}$. 
Thus, two edges $(s, r')$ and $(s, r'')$ are in $M^{post}_{w+1}$, 
which implies $r' = r''$. Therefore, $(s,r')$ is in $E(P^*)$. 
Finally, because $P^*$ is an $M^{pre}_{w+1}$-augmenting path 
and $(s, r') \in E(P^*)$ is in $M^{post}_{w+1}$, by Fact~\ref{fact: s}, 
we then have $\overrightarrow{r', s} \in \overrightarrow{E}(P^*)$.
\end{proof}

For any entry edge $\overrightarrow{r, s}$ of $\bar{P}$, 
we call $s$ an \textbf{entry server} of $\bar{P}$. 
Informally, Claim~\ref{claim: entrynext} states that once 
$\bar{P}$ reaches an entry server, 
$\bar{P}$ starts to traverse $P^*$ in the reverse direction. 
On the other hand, for any exit edge $\overrightarrow{r, s}$ of $\bar{P}$, 
we call $r$ an \textbf{exist request} of $\bar{P}$. 
Because $\bar{P}$ is a path, every vertex appears at most once in $\bar{P}$. 
Thus, once $\bar{P}$ reaches an entry server, 
$\bar{P}$ starts to traverses $P^*$ in the reverse direction until 
$\bar{P}$ reaches an exist request and leaves $P^*$. 
Note that because $\bar{P}$ is an $M^{post}_{w+1}$-augmenting path, 
$\ter(\bar{P})$ is not in $P^*$. 
Thus, there must exist an exist edge after every entry edge.

Let $s^{X}_k(\bar{P})$ be the $k$th entry server of $\bar{P}$, 
and $r^{X}_k(\bar{P})$ be the $k$th exist request of $\bar{P}$. 
Let $cr(\bar{P})$ be the number of entry servers of $\bar{P}$ (or informally, 
the number of times that $\bar{P}$ crosses $P^*$). 
When the $M^{post}_{w+1}$-augmenting path $\bar{P}$ is clear from the context, 
we may simply write $s^X_k$, $r^X_k$, and $cr$. 
$\bar{P}$ can then be written in the form of \linebreak
$\bar{P} = r_i \cdots s^{X}_1 \cdots r^{X}_1 \cdots s^{X}_2 \cdots r^{X}_2 \cdots 
s^{X}_{cr} \cdots r^{X}_{cr} \cdots \ter(\bar{P})$. 
We can then divide $\bar{P}$ into subpaths at entry servers and 
exist requests. Specifically, $\bar{P}$ is divided into 
$H(\bar{P}), B(\bar{P}, 1), W(\bar{P}, 1), B(\bar{P}, 2), W(\bar{P}, 2), \cdots, 
B(\bar{P}, cr-1), W(\bar{P}, cr-1)$, $B(\bar{P}, cr)$, and $T(\bar{P})$, where 
\begin{enumerate}
\item $H(\bar{P}) = r_i \cdots s^{X}_1$ is called the \textbf{head} of $\bar{P}$,
\item $B(\bar{P}, k) = s^{X}_{k} \cdots r^{X}_{k}$ 
      $(1 \leq k \leq cr)$ is called a \textbf{back} of $\bar{P}$, 
\item $W(\bar{P}, k) = r^{X}_{k} \cdots s^{X}_{k+1}$ 
      $(1\leq k \leq cr - 1)$ is called a \textbf{wing} of $\bar{P}$, and
\item $T(\bar{P}) = r^{X}_{cr} \cdots \ter(\bar{P})$  is called the \textbf{tail} of $\bar{P}$.      
\end{enumerate}

\begin{figure}[t]
\caption{An example of the partition of $P$ and the construction of $P'$ where 
$cr = 3$ and $\sgn(s_1^X, r_{cr}^X) = 1$.}
\label{fig: notation}
\begin{center} 
\includegraphics[width=16cm]{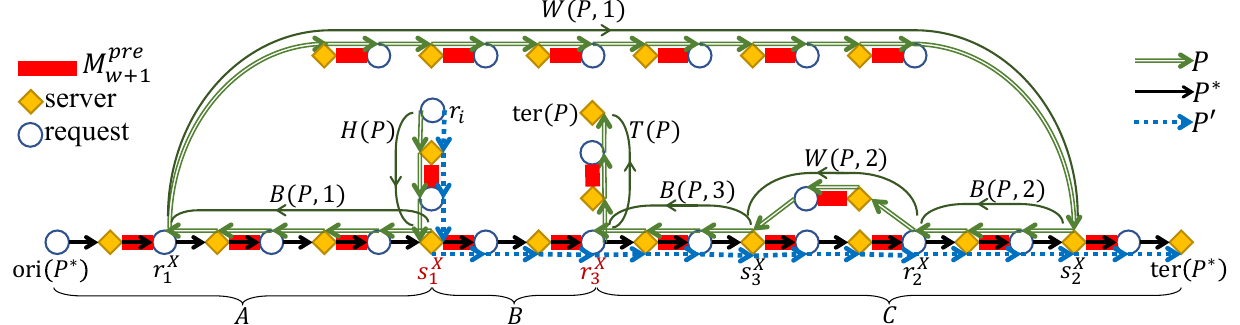} 
\end{center}
\end{figure}

For every path $Q = v_1v_2\cdots v_{\ell}$, 
denote by $\rev(Q)$ the reverse of $Q$, i.e., $\rev(Q) = v_{\ell}v_{\ell-1}\cdots v_1$. 
For every path $Q$, denote by $uQv$ the the subpath of $Q$ that originates  
at $u$ and terminates at $v$. 
Claims~\ref{claim: exist} and \ref{claim: entrynext} imply that 
every back of $\bar{P}$ traverses a subpath of $P^*$ backward. 
Specifically, we have the following simple fact.
\begin{fact}\label{fact: back}
Let $s\bar{P}r$ be any back of $\bar{P}$. 
Then $\rev(s\bar{P}r) = rP^*s$. 
\end{fact}

\paragraph*{Construction of $P'$ and $P''$.}
$P'$ is the concatenation of $H(P)$ and the subpath of $P^*$ that originates 
at $s^{X}_1(P)$ and terminates at $\ter(P^*)$, 
and $P''$ is the concatenation of $H(\widetilde{P})$ and the subpath of $P^*$ 
that originates at $s^{X}_1(\widetilde{P})$ and terminates at $\ter(P^*)$.\footnote{If 
$s^{X}_1(P) = \ter(P^*)$, then $P' = H(P)$. 
Similarly, if $s^{X}_1(\widetilde{P}) = \ter(P^*)$, then $P'' = H(\widetilde{P})$.} 
In other words, $P'$ (respectively, $P''$) first traverses $P$ (respectively, 
$\widetilde{P}$) 
until the first entry server $s^{X}_1(P)$ (respectively, $s^X_1(\widetilde{P})$) 
is met, after which $P'$ (respectively, $P''$) traverses $P^*$. 
See Figs.~\ref{fig: notation} and \ref{fig: notation2} for examples. 
Clearly, $P'$ and $P''$ are $r_i$'s $M^{pre}_{w+1}$-augmenting paths.   

\subsection{Lower Bounds for Backs and Wings}
Observe that while $\varphi(Q)$ is originally defined for any augmenting path $Q$, 
the definition can be extended to any (sub)path $uQv$ that alternates between 
requests and servers:
\[
\varphi(uQv) = \sum_{\overrightarrow{r,s} \in \overrightarrow{E}(uQv)}
                           {\gamma D(r,s)}
          - \sum_{\overrightarrow{s,r} \in \overrightarrow{E}(uQv)}
                           {D(r,s)}.
\]
If $Q$ is a virtual augmenting path terminating at $\widetilde{s}_p$, 
for any subpath $uQ\widetilde{s}_p$ and any time $t \geq a(r_p)$, 
we have 
\[
\varphi_t(uQ\widetilde{s}_p) = \gamma D_t(r_p, \widetilde{s}_p) + \sum_{\overrightarrow{r,s} \in 
                                                          \overrightarrow{E}(uQ\widetilde{s}_p)}
                           {\gamma D(r,s)}
          - \sum_{\overrightarrow{s,r} \in \overrightarrow{E}(uQ\widetilde{s}_p)}
                           {D(r,s)}.
\]
Moreover, if a path $Q$ does not contain any edge, 
then $\varphi(Q) = 0$. By the above definition, we then have  
\[
\varphi(P) = \varphi(H(P)) + \varphi(T(P)) 
         + \sum_{k = 1}^{cr(P)}{\varphi(B(P,k))} 
         + \sum_{k = 1}^{cr(P)-1}{\varphi(W(P,k))}
\]
and
\[
\varphi_{t_{w+1}}(\widetilde{P}) = \varphi(H(\widetilde{P})) + \varphi_{t_{w+1}}(T(\widetilde{P})) 
         + \sum_{k = 1}^{cr(\widetilde{P})}{\varphi(B(\widetilde{P},k))} 
         + \sum_{k = 1}^{cr(\widetilde{P})-1}{\varphi(W(\widetilde{P},k))}.
\]

Let $v^*_h$ be the $h$th vertex in $P^*$. 
Let $|P^*|$ be the number of vertices in $P^*$.
Thus, $P^* = v^*_1v^*_2v^*_3\cdots v^*_{|P^*|}$. 
Moreover, for any vertex $v$ in $P^*$, let $\pi(v)$ be $v$'s order in $P^*$. 
In other words, $\pi(v^*_h) = h$ for any $h \in \{1, 2, \cdots, |P^*|\}$.
For any two vertices $u$ and $v$ in $P^*$, 
define
\[
P^*(u, v) = 
\begin{cases}
&uP^*v, \text{ if } \pi(u) < \pi(v)\\
&vP^*u, \text{ if } \pi(u) > \pi(v)
\end{cases}
\] 
and
\[
\sgn(u, v) = 
\begin{cases}
&1, \text{ if } \pi(u) < \pi(v)\\
&-1, \text{ if } \pi(u) > \pi(v).
\end{cases}
\] 
Thus, $P^*(u, v)$ gives the subpath of $P^*$ between $u$ and $v$ 
in the \textit{correct} direction, and $\sgn(u,v)$ outputs 1 if $u$ precedes
$v$ in $P^*$ and outputs $-1$ otherwise. 

The next lemma lower bounds the $\gamma$-net-cost of backs and wings.
\begin{lemma}\label{lemma: 230}
Let $\bar{P}$ be any $M^{post}_{w+1}$-augmenting path of $r_i$. Let $cr = cr(\bar{P})$. Then
\[
\sum_{k = 1}^{cr}{\varphi(B(\bar{P},k))} 
         + \sum_{k = 1}^{cr-1}{\varphi(W(\bar{P},k))} \geq 
         \sgn(s^X_1,r^X_{cr})\varphi(P^*(s^X_1,r^X_{cr})).         
\]
\end{lemma}

\begin{figure}[t]
\caption{An example of the partition of $P$ and the construction of $P'$ 
where $cr = 2$ and $\sgn(s_1^X, r_{cr}^X) = -1$.}
\label{fig: notation2}
\begin{center} 
\includegraphics[width=16cm]{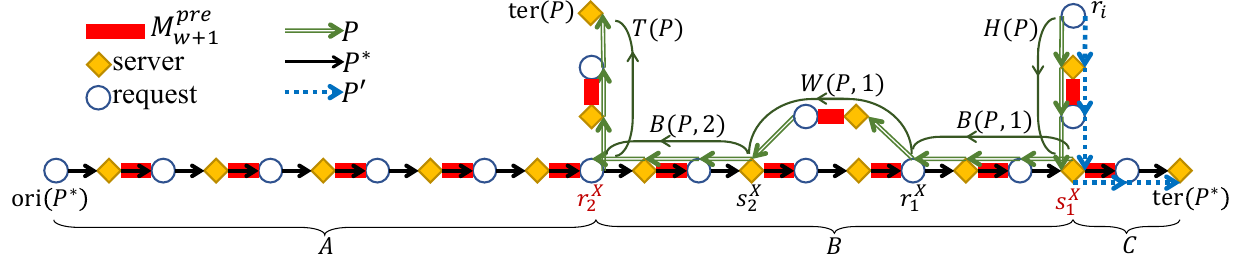} 
\end{center}
\end{figure}

To prove Lemma~\ref{lemma: 230}, we first prove the following 
transitive property. 
\begin{claim}\label{claim: tras}
Let $v_1$, $v_2$, and $v_3$ be any three distinct vertices in $P^*$. 
Then
\[
\sgn(v_1,v_2)\varphi(P^*(v_1,v_2))+\sgn(v_2,v_3)\varphi(P^*(v_2,v_3)) 
= \sgn(v_1,v_3)\varphi(P^*(v_1,v_3)).
\] 
\end{claim} 
\begin{proof}
We divide the proof into cases according to the ordering of $\pi(v_1)$, $\pi(v_2)$, and $\pi(v_3)$.

\paragraph{Case 1: $\pi(v_1) < \pi(v_2) < \pi(v_3)$:} 
In this case, 
\begin{align*}
&\sgn(v_1,v_2)\varphi(P^*(v_1,v_2))+\sgn(v_2,v_3)\varphi(P^*(v_2,v_3))\\ 
&= \varphi(P^*(v_1,v_2)) + \varphi(P^*(v_2,v_3))
= \varphi(P^*(v_1,v_3)) = \sgn(v_1,v_3)\varphi(P^*(v_1,v_3)).
\end{align*}

\paragraph{Case 2: $\pi(v_1) < \pi(v_3) < \pi(v_2) $:} 
In this case, 
\begin{align*}
&\sgn(v_1,v_2)\varphi(P^*(v_1,v_2))+\sgn(v_2,v_3)\varphi(P^*(v_2,v_3))\\ 
&= \varphi(P^*(v_1,v_2)) - \varphi(P^*(v_2,v_3))
= \varphi(P^*(v_1,v_3)) = \sgn(v_1,v_3)\varphi(P^*(v_1,v_3)).
\end{align*}

\paragraph{Case 3: $\pi(v_3) < \pi(v_1) < \pi(v_2) $:} 
In this case, 
\begin{align*}
&\sgn(v_1,v_2)\varphi(P^*(v_1,v_2))+\sgn(v_2,v_3)\varphi(P^*(v_2,v_3))\\ 
&= \varphi(P^*(v_1,v_2)) - \varphi(P^*(v_2,v_3))
= -\varphi(P^*(v_1,v_3)) = \sgn(v_1,v_3)\varphi(P^*(v_1,v_3)).
\end{align*}

\paragraph{Case 4: $\pi(v_2) < \pi(v_1) < \pi(v_3)$:} 
In this case, 
\begin{align*}
&\sgn(v_1,v_2)\varphi(P^*(v_1,v_2))+\sgn(v_2,v_3)\varphi(P^*(v_2,v_3)) \\
&= -\varphi(P^*(v_1,v_2)) +\varphi(P^*(v_2,v_3))
= \varphi(P^*(v_1,v_3)) = \sgn(v_1,v_3)\varphi(P^*(v_1,v_3)).
\end{align*}

\paragraph{Case 5: $\pi(v_2) < \pi(v_3) < \pi(v_1) $:} 
In this case, 
\begin{align*}
&\sgn(v_1,v_2)\varphi(P^*(v_1,v_2))+\sgn(v_2,v_3)\varphi(P^*(v_2,v_3)) \\
&= -\varphi(P^*(v_1,v_2)) +\varphi(P^*(v_2,v_3))
= -\varphi(P^*(v_1,v_3)) = \sgn(v_1,v_3)\varphi(P^*(v_1,v_3)).
\end{align*}

\paragraph{Case 6: $\pi(v_3) < \pi(v_2) < \pi(v_1) $:} 
In this case, 
\begin{align*}
&\sgn(v_1,v_2)\varphi(P^*(v_1,v_2))+\sgn(v_2,v_3)\varphi(P^*(v_2,v_3)) \\
&= -\varphi(P^*(v_1,v_2)) -\varphi(P^*(v_2,v_3))
= -\varphi(P^*(v_1,v_3)) = \sgn(v_1,v_3)\varphi(P^*(v_1,v_3)).
\end{align*}
\end{proof}

Next, we lower bound $\varphi(B(\bar{P}, k))$. 
\begin{claim} \label{claim: backlb}
Let $s\bar{P}r$ be any back of $\bar{P}$. 
If $\gamma \geq 1$, then $\varphi_{\gamma}(s\bar{P}r) \geq \sgn(s,r)\varphi_{\gamma}(P^*(s,r))$. 
\end{claim}

\begin{proof}
By Fact~\ref{fact: back}, it suffices to prove $\varphi_{\gamma}(s\bar{P}r) \geq -\varphi_{\gamma}(rP^*s)$ 
or equivalently, $-\varphi_{\gamma}(s\bar{P}r) \leq \varphi_{\gamma}(rP^*s)$.
\begin{align*}
-\varphi_{\gamma}(s\bar{P}r) 
           &= -\sum_{\overrightarrow{r_p, s_q} \in \overrightarrow{E}(s\bar{P}r)}
                           {\gamma D(r_p,s_q)}
            +\sum_{\overrightarrow{s_q, r_p} \in \overrightarrow{E}(s\bar{P}r)}
                           {D(r_p, s_q)}\\
           &\leq -\sum_{\overrightarrow{r_p, s_q} \in \overrightarrow{E}(s\bar{P}r)}
                           {D(r_p,s_q)}
            +\sum_{\overrightarrow{s_q, r_p} \in \overrightarrow{E}(s\bar{P}r)}
                           {\gamma D(r_p, s_q)}\\
           &= -\sum_{\overrightarrow{s_q, r_p} \in \overrightarrow{E}(\rev(s\bar{P}r))}
                           {D(r_p,s_q)} 
              +\sum_{\overrightarrow{r_p, s_q} \in \overrightarrow{E}(\rev(s\bar{P}r))}
                           {\gamma D(r_p, s_q)} \\          
           &=\varphi_{\gamma}(\rev(s\bar{P}r)) = \varphi_{\gamma}(rP^*s),
\end{align*}
where the last equality is due to Fact~\ref{fact: back}.
\end{proof}

Next, we lower bound $\varphi(W(\bar{P}, k))$. 
\begin{claim}\label{claim: fwlb}
Let $r\bar{P}s$ be any wing of $\bar{P}$ such that $\sgn(r,s) = 1$. 
Then $\varphi(r\bar{P}s) \geq \sgn(r,s)\varphi(P^*(r,s))$. 
\end{claim}
\begin{proof}
Observe that edges in wings are not in $E(P^*)$. 
Thus, $r\bar{P}s$ alternates between edges in $M^{pre}_{w+1}$ 
and edges not in $M^{pre}_{w+1}$ (see $W(P,1)$ in 
Fig.~\ref{fig: notation} for an example). 
Because $\sgn(r,s) = 1$, 
replacing $P^*(r,s)$ with $r\bar{P}s$ in $P^*$ yields 
another $M^{pre}_{w+1}$-augmenting path $Q$ that originates at $\ori(P^*)$. 
Because $P^*$ is $\ori(P^*)$'s minimum $M^{pre}_{w+1}$-augmenting path, we have
$\varphi(Q) \geq \varphi(P^*)$, which implies 
$\varphi(r\bar{P}s) \geq \varphi(P^*(r,s))$. 
\end{proof}

For the case where $\sgn(r,s) = -1$, we need the following lemma, 
whose proof is similar to that of Lemma~\ref{lemma: diffM}.
\begin{lemma}\label{lemma: diffM2}
Let $M_1 \subseteq M^{OFF}$. Let $M_2$ be any matching that saturates the same set of 
vertices as $M_1$. Then $D(M_1) \leq \gamma D(M_2)$.
\end{lemma}
\begin{proof}
By Invariant~\eqref{eq: dualrelaxed}, we have 
\[
\gamma \cdot D(M_2) = \sum_{(r,s)\in M_2}{\gamma \cdot D(r,s)} 
                    \geq \sum_{(r,s)\in M_2}{(z(r)+z(s))} = \sum_{(r,s)\in M_1}{(z(r)+z(s))},
\]
where the last inequality holds because $M_1$ and $M_2$ saturate the same set of vertices. 
Because $M_1 \subseteq M^{OFF}$ and by Invariant~\eqref{eq: dualtight}, 
we have
\[
\sum_{(r,s)\in M_1}{(z(r)+z(s))} = \sum_{(r,s)\in M_1}{D(r,s)} = D(M_1).
\]
\end{proof}

\begin{claim}\label{claim: bwlb}
Let $r\bar{P}s$ be any wing of $\bar{P}$ such that $\sgn(r,s) = -1$. 
Then $\varphi(r\bar{P}s) \geq \sgn(r,s)\varphi(P^*(r,s))$.
\end{claim}
\begin{proof}
We prove $\varphi(r\bar{P}s) + \varphi(P^*(r,s)) \geq 0$. 
Let 
\[
M_1 = 
\{(s_q,r_p) | 
              \overrightarrow{s_q,r_p} \in 
                  \overrightarrow{E}(r\bar{P}s)
                    \cup
                  \overrightarrow{E}(P^*(r,s))   
        \}
\]
and
\[
M_2 = 
\{(r_p,s_q) | 
              \overrightarrow{r_p,s_q} \in 
                  \overrightarrow{E}(r\bar{P}s)
                    \cup
                  \overrightarrow{E}(P^*(r,s))   
        \}.
\]
Observe that $M_1 \subseteq M^{pre}_{w+1}$. 
In addition, $M_1$ and $M_2$ saturate the same set of vertices because $r\bar{P}s$ and $P^*(r,s)$ 
form an alternating cycle (see $W(P,2)$ in Fig.~\ref{fig: notation} for an example). 
By Lemma~\ref{lemma: diffM2}, 
$D(M_1) \leq \gamma D(M_2)$.
Because $r\bar{P}s$ and $P^*(r,s)$ do not share edges, 
\[
\varphi(r\bar{P}s) + \varphi(P^*(r,s)) 
= \sum_{(r_p, s_q) \in M_2}{\gamma D(r_p, s_q)} 
 -\sum_{(s_q, r_p) \in M_1}{D(r_p, s_q)}
=\gamma D(M_2) - D(M_1) \geq 0.
\]
\end{proof}

\paragraph*{Proof of Lemma~\ref{lemma: 230}.}
By Claims~\ref{claim: backlb}, \ref{claim: fwlb}, and \ref{claim: bwlb}, we have 
\begin{align*}
    &\sum_{k = 1}^{cr}{\varphi(B(\bar{P},k))} 
   + \sum_{k = 1}^{cr-1}{\varphi(W(\bar{P},k))} 
=   \sum_{k = 1}^{cr}{\varphi(s^X_k\bar{P}r^X_k)} 
   + \sum_{k = 1}^{cr-1}{\varphi(r^X_k\bar{P}s^X_{k+1})} \\
\geq 
    &\sum_{k = 1}^{cr}{\sgn(s^X_k, r^X_k)\varphi(P^*(s^X_k, r^X_k))} 
   + \sum_{k = 1}^{cr-1}{\sgn(r^X_k, s^X_{k+1})\varphi(P^*(r^X_k, s^X_{k+1}))}
\stackrel{\text{Claim~\ref{claim: tras}}}{=}     \sgn(s^X_1, r^X_{cr})\varphi(P^*(s^X_1, r^X_{cr})).
\end{align*}
  
\subsection{Proof of Eq.~\eqref{eq: P'goal}}
In the following proof, $s^X_k, r^X_k$, and $cr$ refer to 
$s^X_k(P), r^X_k(P)$, and $cr(P)$, respectively. 
We first consider the case where $\sgn(s^X_1, r^X_{cr}) = 1$. 
Divide $P^*$ into $A, B$, and $C$ such that
$A = \ori(P^*)P^*s^X_{1}$, 
$B = s^X_{1}P^*r^X_{cr}$, 
and
$C = r^X_{cr}P^*\ter(P^*)$.
Observe that $\varphi(P^*) = \varphi(A)+\varphi(B)+\varphi(C)$ 
and $P'$ is the concatenation of $H(P)$, $B$, and $C$ 
(see Fig.~\ref{fig: notation} for an example). 
Moreover, the concatenation of $A$, $B$, and $T(P)$ is an $M^{pre}_{w+1}$-augmenting path that originates 
at $\ori(P^*)$. Because $P^*$ is $\ori(P^*)$'s real minimum $M^{pre}_{w+1}$-augmenting path, we then have 
\begin{align*}
\varphi(P^*) = \varphi(A)+\varphi(B)+\varphi(C) 
            & \leq \varphi(A)+\varphi(B)+\varphi(T(P))\\
\varphi(H(P))+\varphi(B)+\varphi(C) & \leq \varphi(H(P))+\varphi(B)+\varphi(T(P))\\
                        \varphi(P') &\leq \varphi(H(P))+\varphi(s^X_{1}P^*r^X_{cr})+\varphi(T(P))\\
                        \varphi(P') &\stackrel{\text{Lemma~\ref{lemma: 230}}}{\leq}    
                                      \varphi(H(P))+\sum_{k = 1}^{cr}{\varphi(B(P,k))} 
                                      + \sum_{k = 1}^{cr-1}{\varphi(W(P,k))}+\varphi(T(P))  \\
                        \varphi(P') &\leq \varphi(P).             
\end{align*}

Next, consider the case where $\sgn(s^X_1, r^X_{cr}) = -1$. 
Divide $P^*$ into $A, B$, and $C$ such that
$A = \ori(P^*)P^*r^X_{cr}$, 
$B = r^X_{cr}P^*s^X_{1}$, 
and
$C = s^X_{1}P^*\ter(P^*)$.
Observe that $\varphi(P^*) = \varphi(A)+\varphi(B)+\varphi(C)$ 
and $P'$ is the concatenation of $H(P)$ and $C$ 
(see Fig.~\ref{fig: notation2} for an example). 
Moreover, the concatenation of $A$ and $T(P)$ is an $M^{pre}_{w+1}$-augmenting path that originates 
at $\ori(P^*)$. Because $P^*$ is $\ori(P^*)$'s real minimum $M^{pre}_{w+1}$-augmenting path, we then have 
\begin{align*}
\varphi(P^*) = 
\varphi(A)   +\varphi(B)+\varphi(C) & \leq \varphi(A)+\varphi(T(P))\\
\varphi(H(P))+\varphi(B)+\varphi(C) & \leq \varphi(H(P))+\varphi(T(P))\\
\varphi(H(P))+\varphi(C) &\leq \varphi(H(P))-\varphi(B) +\varphi(T(P))\\
                  \varphi(P') &\leq \varphi(H(P))-\varphi(r^X_{cr}P^*s^X_{1}) +\varphi(T(P))\\
                  \varphi(P') &\stackrel{\text{Lemma~\ref{lemma: 230}}}{\leq}    
                                      \varphi(H(P))+\sum_{k = 1}^{cr}{\varphi(B(P,k))} 
                                      + \sum_{k = 1}^{cr-1}{\varphi(W(P,k))}+\varphi(T(P))  \\
                        \varphi(P') &\leq \varphi(P).             
\end{align*}

\subsection{Proof of Eq.~\eqref{eq: P''goal}}
In the following proof, $s^X_k, r^X_k$, and $cr$ refer to 
$s^X_k(\widetilde{P}), r^X_k(\widetilde{P})$, and $cr(\widetilde{P})$, respectively. 
The proof is similar to that of Eq.~\eqref{eq: P'goal}. 
The main difference is that we use the design of the VRM algorithm 
that when a request (e.g. $\ori(P^*) = r^*$) 
becomes ready, its virtual minimum $\gamma$-net-cost 
is at least its real minimum $\gamma$-net-cost.

We first consider the case where $\sgn(s^X_1, r^X_{cr}) = 1$. 
Divide $P^*$ into $A, B$, and $C$ such that
$A = \ori(P^*)P^*s^X_{1}$, 
$B = s^X_{1}P^*r^X_{cr}$, 
and
$C = r^X_{cr}P^*\ter(P^*)$.
Observe that $\varphi(P^*) = \varphi(A)+\varphi(B)+\varphi(C)$ 
and $P''$ is the concatenation of $H(P)$, $B$, and $C$. 
Let $\widetilde{Q}$ be the concatenation of $A$, $B$, and $T(\widetilde{P})$.  
Then $\widetilde{Q}$ is a virtual $M^{pre}_{w+1}$-augmenting path that originates at $\ori(P^*)$. 
Because $P^*$ is $\ori(P^*)$'s real minimum $M^{pre}_{w+1}$-augmenting path and 
$r^*$ becomes ready at time $t_{w+1}$, we then have $\varphi(P^*) \leq \varphi_{t_{w+1}}(\widetilde{Q})$ 
and thus 
\begin{align*}
\varphi(A)+\varphi(B)+\varphi(C) &= \varphi(P^*) \leq \varphi_{t_{w+1}}(\widetilde{Q})   
            = \varphi(A)+\varphi(B)+\varphi_{t_{w+1}}(T(\widetilde{P}))\\
\varphi(H(\widetilde{P}))+\varphi(B)+\varphi(C) & 
                    \leq \varphi(H(\widetilde{P}))+\varphi(B)+\varphi_{t_{w+1}}(T(\widetilde{P}))\\
      \varphi(P'') &\leq \varphi(H(\widetilde{P}))+\varphi(s^X_{1}P^*r^X_{cr})+\varphi_{t_{w+1}}(T(\widetilde{P}))\\
      \varphi(P'') &\stackrel{\text{Lemma~\ref{lemma: 230}}}{\leq}    
                                      \varphi(H(\widetilde{P}))+\sum_{k = 1}^{cr}{\varphi(B(\widetilde{P},k))} 
                                      + \sum_{k = 1}^{cr-1}{\varphi(W(\widetilde{P},k))}+\varphi_{t_{w+1}}(T(\widetilde{P}))  \\
      \varphi(P'') &\leq \varphi_{t_{w+1}}(\widetilde{P}).             
\end{align*}

Next, consider the case where $\sgn(s^X_1, r^X_{cr}) = -1$. 
Divide $P^*$ into $A, B$, and $C$ such that
$A = \ori(P^*)P^*r^X_{cr}$, 
$B = r^X_{cr}P^*s^X_{1}$, 
and
$C = s^X_{1}P^*\ter(P^*)$.
Observe that $\varphi(P^*) = \varphi(A)+\varphi(B)+\varphi(C)$ 
and $P''$ is the concatenation of $H(\widetilde{P})$ and $C$. 
Let $\widetilde{Q}$ be the the concatenation of $A$ and $T(\widetilde{P})$.
Then $\widetilde{Q}$ is an $M^{pre}_{w+1}$-augmenting path that originates 
at $\ori(P^*)$. Because $P^*$ is $\ori(P^*)$'s real minimum $M^{pre}_{w+1}$-augmenting path and $r^*$ becomes 
ready at time $t_{w+1}$, we then have $\varphi(P^*) \leq \varphi_{t_{w+1}}(\widetilde{Q})$ 
and thus 
\begin{align*}
\varphi(A)   +\varphi(B)+\varphi(C) & = \varphi(P^*) 
  \leq \varphi_{t_{w+1}}(\widetilde{Q}) = \varphi(A)+\varphi_{t_{w+1}}(T(\widetilde{P}))\\
\varphi(H(\widetilde{P}))+\varphi(B)+\varphi(C) 
& \leq \varphi(H(\widetilde{P}))+\varphi_{t_{w+1}}(T(\widetilde{P}))\\
\varphi(H(\widetilde{P}))+\varphi(C) 
&\leq \varphi(H(\widetilde{P}))-\varphi(B) +\varphi_{t_{w+1}}(T(\widetilde{P}))\\
\varphi(P'') &\leq \varphi(H(\widetilde{P}))-\varphi(r^X_{cr}P^*s^X_{1})+\varphi_{t_{w+1}}(T(\widetilde{P}))\\
\varphi(P'') &\stackrel{\text{Lemma~\ref{lemma: 230}}}{\leq}    
                                      \varphi(H(\widetilde{P}))+\sum_{k = 1}^{cr}{\varphi(B(\widetilde{P},k))} 
                                      + \sum_{k = 1}^{cr-1}{\varphi(W(\widetilde{P},k))}+\varphi_{t_{w+1}}(T(\widetilde{P}))  \\
                        \varphi(P'') &\leq \varphi_{t_{w+1}}(\widetilde{P}).             
\end{align*}

\end{document}